\documentclass[10pt]{article}

\usepackage{fullpage}
\usepackage{tikz}
\usetikzlibrary{arrows,automata,shapes,decorations,calc,matrix,decorations.pathmorphing}
\usepackage{graphicx}
\usepackage{caption}
\usepackage{subcaption}
\usepackage{xcolor}

\usepackage{authblk}

\usepackage{amsfonts,amsmath,amsthm}
\newtheorem{theorem}{Theorem}
\theoremstyle{definition}
\newtheorem{lemma}{Lemma}

\def\AT{\operatorname{AT}}
\def\ET{\operatorname{ET}}
\def\CT{T}
\def\fp{\operatorname{fp}}
\def\E{\mathbb{E}}
\def\Em{\mathcal{M}}
\def\El{\mathcal{L}}

\def\temp{\mathcal{T}}
\def\eps{\alpha}
\def\sne{\subsubsection*}
\def\chapter{\relax}

\title{Fixation probability and fixation time in structured populations}
\date{}

\author[a, 1]{Josef\! Tkadlec}
\author[b, 1]{Andreas\! Pavlogiannis}
\author[a]{Krishnendu\! Chatterjee}
\author[c, $\star$]{Martin\! A.\! Nowak}
\affil[a]{IST Austria, A-3400 Klosterneuburg, Austria}
\affil[b]{Lab for Automated Reasoning and Analysis, EPFL, CH-1015 Lausanne, Switzerland}
\affil[c]{Program for Evolutionary Dynamics, Department of Organismic and Evolutionary Biology, Department of Mathematics, Harvard University, Cambridge, MA 02138, USA}
\affil[1]{J.T. and A.P. contributed equally to this work.}
\affil[$\star$]{To whom correspondence should be addressed. E-mail: martin\_nowak@harvard.edu}

\begin{document}

\maketitle

\tableofcontents

\begin{abstract}
The rate of biological evolution depends on the fixation probability and on the fixation time of new mutants.
Intensive research has focused on identifying population structures that augment the fixation probability of advantageous mutants.
But these ``amplifiers of natural selection'' typically increase fixation time. 
Here we study population structures that achieve a trade-off between high fixation probability and short fixation time.
First, we show that no amplifiers can have asymptotically lower absorption time than the well-mixed population. 
Then we design population structures that substantially augment the fixation probability with just a minor increase in fixation time.
Finally, we show that those structures enable higher effective rate of evolution than the well-mixed population provided that the rate of generating advantageous mutants is relatively low.
Our work sheds light on how population structure affects the rate of evolution.
Moreover, our structures could be useful for lab-based, medical or industrial applications of evolutionary optimization.
\end{abstract}

\section{Introduction}

The two primary forces that drive evolutionary processes are mutation and selection.
Mutation generates new variants in a population.
Selection chooses among them depending on the reproductive rates of individuals.
Evolutionary processes are intrinsically random. A new mutant that is initially present in the population at low frequency can go extinct due to random drift. 
The key quantities of evolutionary dynamics which affect the rate of evolution are~\cite{Kimura68,Ewens04,Nowak2006,Desai07,McCandlish15}: 
(a)~the mutation rate $\mu$, which is the rate at which new mutants are generated;
(b)~the fixation probability $\rho$, which is the probability that the lineage of a mutant takes over the whole population; and 
(c)~the fixation time $\tau$, which is the expected time 
until the lineage of a mutant fixates in the population.

A classical and well-studied evolutionary process is the discrete-time Moran birth-death process~\cite{Moran1962}. 
Given a population of $N$ individuals, at each time step an individual is chosen for reproduction proportionally to its fitness; then the offspring replaces a random individual (see Figure~\ref{fig:1}a).
In the case of a well-mixed population, each offspring is equally likely to replace any other individual.
For a single new mutant with relative fitness $r$, its fixation probability is $\rho=(1-1/r)/(1-1/r^{N})$.
Thus, for $r>1$ and large $N$ we have $\rho \approx 1-1/r$~\cite{lieberman2005,Nowak2006}.

For measuring time, there are two natural options.
The \textit{absorption time} is the average number of steps of the Moran process until the population becomes homogeneous, regardless of whether the mutant fixates or becomes extinct.
Alternatively, the (conditional) \textit{fixation time} is the average number of steps of those evolutionary trajectories that lead to the fixation of the mutant, ignoring trajectories that lead to the extinction of the mutant.
Since the evolutionary trajectories leading to extinction are typically shorter than those leading to fixation, the fixation time tends to be longer than the absorption time.
Therefore, in our results concerning time we present lower bounds on the absorption time and upper bounds on the fixation time.

For the well-mixed population, both the absorption time and the fixation time are of order of $N \log N$~\cite{diaz2016absorption, altrock2009fixation}.
Specifically, for $r>1$ and large $N$, the absorption time is approximately $\frac{r+1}{r}\cdot N \log N$
while the fixation time is approximately $\frac{r+1}{r-1}\cdot N \log N$.
For neutral evolution, $r=1$, the absorption time is approximately  $N\log N$ while the fixation time is $(N-1)^2$.

Both the fixation probability and the fixation time depend on population structure~\cite{Slatkin81,Nowak1992,durrett1994stochastic,Whitlock2003,Hauert04,komarova2006spatial,Houchmandzadeh2011,Frean2013,komarova2014complex}.
Evolutionary graph theory is a framework to study the effect of population structure.
In evolutionary graph theory, the structure of a population is represented by a graph~\cite{lieberman2005,Broom2008,Broom2011,shakarian12,Debarre14,Allen2017}: 
each individual occupies a vertex; the edges represent the connections to neighboring sites where a reproducing individual can place an offspring. 
The edge weights represent the proportional preference to make such a choice.
The well-mixed population is given by the complete graph $K_N$ where each individual is connected to each other individual (Figure~\ref{fig:1}b). Graphs can also represent deme structured populations, where islands are represented by complete graphs and connections (of different weights) exist between islands. Graphs can also represent spatial lattices or asymmetric structures.

A well-studied example is the star graph $S_N$, which has one central vertex and $N-1$ surrounding vertices each connected to the central vertex (Figure~\ref{fig:1}b).
For the star graph, the fixation probability tends to approximately $1-1/r^2$ for $r>1$ and large $N$
while both the absorption and the fixation time is of the order of $N^2\log N$~\cite{Chalub16,broom2011stars,askari2015analytical}. 
Hence, if a mutant has $10\ \%$ fitness advantage, which means $r=1.1$, the star graph amplifies the advantage to $21\ \%$,
but at the cost of increasing the time to fixation (Figure~\ref{fig:1}c).

%
\begin{figure}
\includegraphics[width=\linewidth]{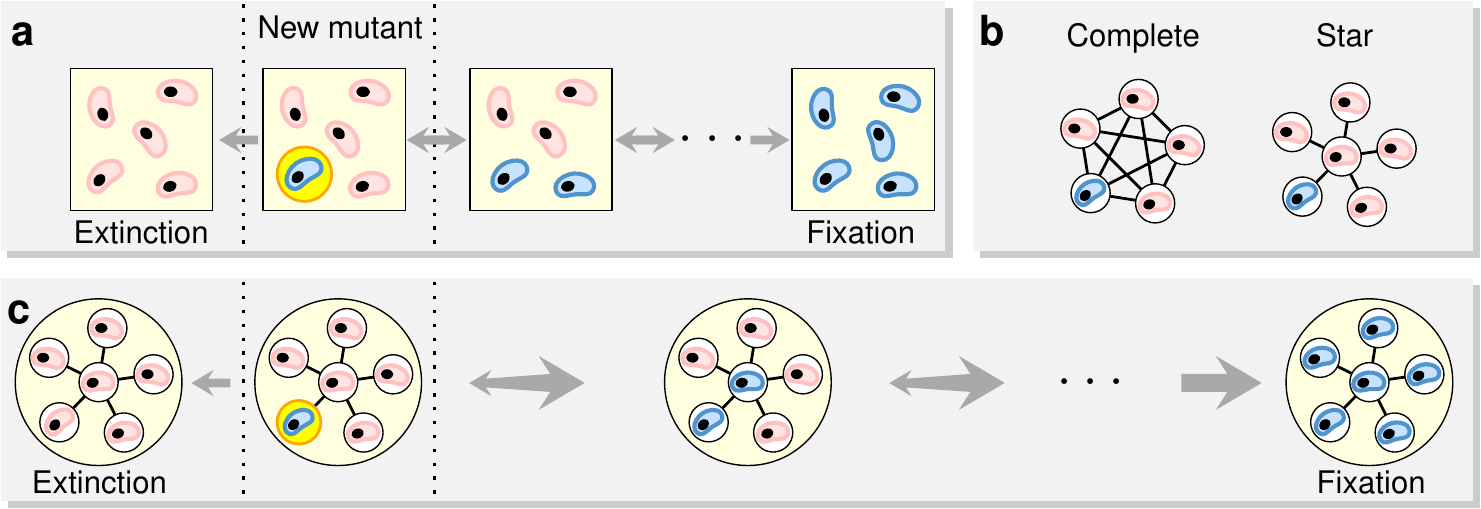}
\caption{
{\bf Moran process on graphs.}
\textbf{a,} A new mutant (blue) appears in a population of finite size. The lineage of the new mutant can either become extinct or reach fixation.  The Moran process is a birth-death process; in any one time step one new offspring is generated and one individual dies.
\textbf{b,} All fixed spatial structures can be described by graphs. The classical, well-mixed population corresponds to a complete graph, where all positions are equivalent.  The star graph is a well-studied example of extreme heterogeneity, where one individual, the center, is connected to all others, but each leaf is only connected to the center.
\textbf{c,} Population structure influences both the fixation probability and the fixation time. An advantageous mutant introduced at a random vertex of a star graph is more likely to fixate than on a complete graph (the arrows pointing to the right are thicker), but the (average) fixation time on the star graph is much longer than on the complete graph (the arrows are longer). The star graph achieves amplification at the cost of deceleration.
}
\label{fig:1}
\end{figure}

Several population structures have been identified that alter the fixation probability of advantageous mutants.
Structures that decrease the fixation probability are known as {\em suppressors of selection} and those that increase it are known as {\em amplifiers of selection}~\cite{lieberman2005,ACN15,HT15,beatingstar}.
However, amplification is usually achieved at the cost of increasing fixation time~\cite{Whitlock2003,Frean2013,diaz2014approximating,Hindersin2014}. 
For example, the star graph has higher fixation probability but also longer fixation time as compared to the well-mixed population.
There also exist {\em superamplifiers} (also known as arbitrarily strong amplifiers of natural selection) 
that guarantee fixation of advantageous mutants in the limit of large population size~\cite{Galanis17,Giakkoupis16,Goldberg16,pavlogiannis2018construction}.
But those structures tend to require even longer fixation times.

We can refer to population structures that decrease the fixation time with respect to the well-mixed population as {\em accelerators}.
Both the fixation probability and the fixation time play an important role in the speed of evolution.
Ideally, we prefer a population structure that is both an amplifier and an accelerator,
but all known amplifiers achieve amplification at the cost of deceleration.
In fact, this slowdown can be so prominent that it outweighs the amplification and leads to longer evolutionary timescales~\cite{Frean2013}.

Here we show that absorption time on any amplifier is asymptotically at least as large as both the absorption and the fixation time on the well-mixed population.
Given this negative result, we proceed to study the trade-off between fixation probability and time more closely.
We have computed fixation probabilities and fixation times for a large class of graphs.
While within this class, the well-mixed population is optimal with respect to fixation time, and the star graph is favorable with respect to fixation probability,
there is a very interesting trade-off curve between fixation probability and fixation time.
In other words, there exist population structures which provide different trade-offs between high fixation probability and short fixation time.
As our main analytical results, we present population structures that asymptotically achieve
fixation probability equal to that of star graphs and fixation time similar to that of well-mixed populations.
Thus, we achieve amplification with negligible deceleration.

Finally, while the above analytical results are established for large population sizes, we also study evolutionary processes on population structures of small or intermediate size by numerical simulation.
Specifically, we consider the effective rate of evolution as proposed by Frean, Rainey, and Traulsen~\cite{Frean2013}.
Generally speaking, the well-mixed population has a high effective rate of evolution if the mutation rate is high, while the star graph has a high effective rate of evolution if the mutation rate is very low.
We show that for a wide range of intermediate mutation rates, our new structures achieve higher effective rate of evolution than both the well-mixed population and the star graph.

\section{Results}

We study several fundamental questions related to the probability-time trade-off 
of a single advantageous mutant in a population of size $N$.
Mutants can arise either spontaneously or during reproduction.
Mutants that arise spontaneously appear at a vertex chosen uniformly at random among all $N$ vertices. This is called {\em uniform initialization}. Mutants that arise during reproduction appear at each vertex
proportionally to its replacement rate, which is called \textit{temperature} of that vertex. This is called {\em temperature initialization}.
We study the probability-time trade-off for both types of initialization.

\subsection{Amplifiers and accelerators}
First, we investigate whether there are population structures that are amplifiers and asymptotic accelerators of selection as compared to the complete graph (well-mixed population).
We show that for any amplifier with population size $N$, the absorption time is of the order of at least
$ N\log N$, for both types of initialization.
Since the fixation time tends to be even longer than the absorption time and the fixation time for the complete graph is also of the order of $N\log N$, regardless of initialization,
this suggests that no amplifier is an asymptotic accelerator.
Moreover, we show that the same conclusion holds for graphs that decrease the fixation probability by not more than a constant factor.
Our result doesn't completely exclude the possibility of population structures with absorption time asymptotically shorter than that of the complete graphs but it shows that, if such structures do exist, then the fixation probability has to tend to 0 as the population size $N$ grows large.
While the above results holds in the limit of large population size, we present a small directed graph that is a 
suppressor and has a slightly shorter fixation time than the complete graph for the same population size (see  Appendix 1).

\subsection{Uniform initialization: $\eps$-Balanced bipartite graphs}
Second, we consider uniform initialization. 
There are two interesting questions:
(1) For fixed (small) population size, how do different population structures fare
with respect to the probability-time trade-off?
(2) In the limit of large population size, do there exists population structures that
achieve the same amplification as the star graph, with shorter fixation time?
Our results are as follows:
\begin{itemize}
\item[(1)] For small population size, both fixation probability and fixation time can be computed numerically~\cite{hindersin2016exact}. 
 We do this for all graphs with $N=8$ vertices and a mutant with relative fitness advantage $r=1.1$ (see Figure~\ref{fig:2a}).
We observe that the complete graph has the shortest fixation time, and the star graph has the highest fixation probability.
However, the star graph has much longer fixation time than the complete graph.
While some graphs have smaller fixation probability and longer fixation time than the complete graph,
there are other graphs which provide a trade-off between high fixation probability and short fixation time.
In particular, there are Pareto-optimal graphs. Recall that in two- or multi-dimensional optimization problems, the Pareto front is the set of non-dominated objects. In our case, the Pareto front consists of graphs for which the fixation probability can not be improved without increasing the fixation time.
For $N=8$ and $r=1.1$ the complete graph and the star graph are the two extreme points of the Pareto front. This finding holds for other values of $r>1$ and $N$ as well (see Figures~\ref{figsupp:sf1} and~\ref{figsupp:sf2}).

%
\begin{figure}
\includegraphics[width=\linewidth]{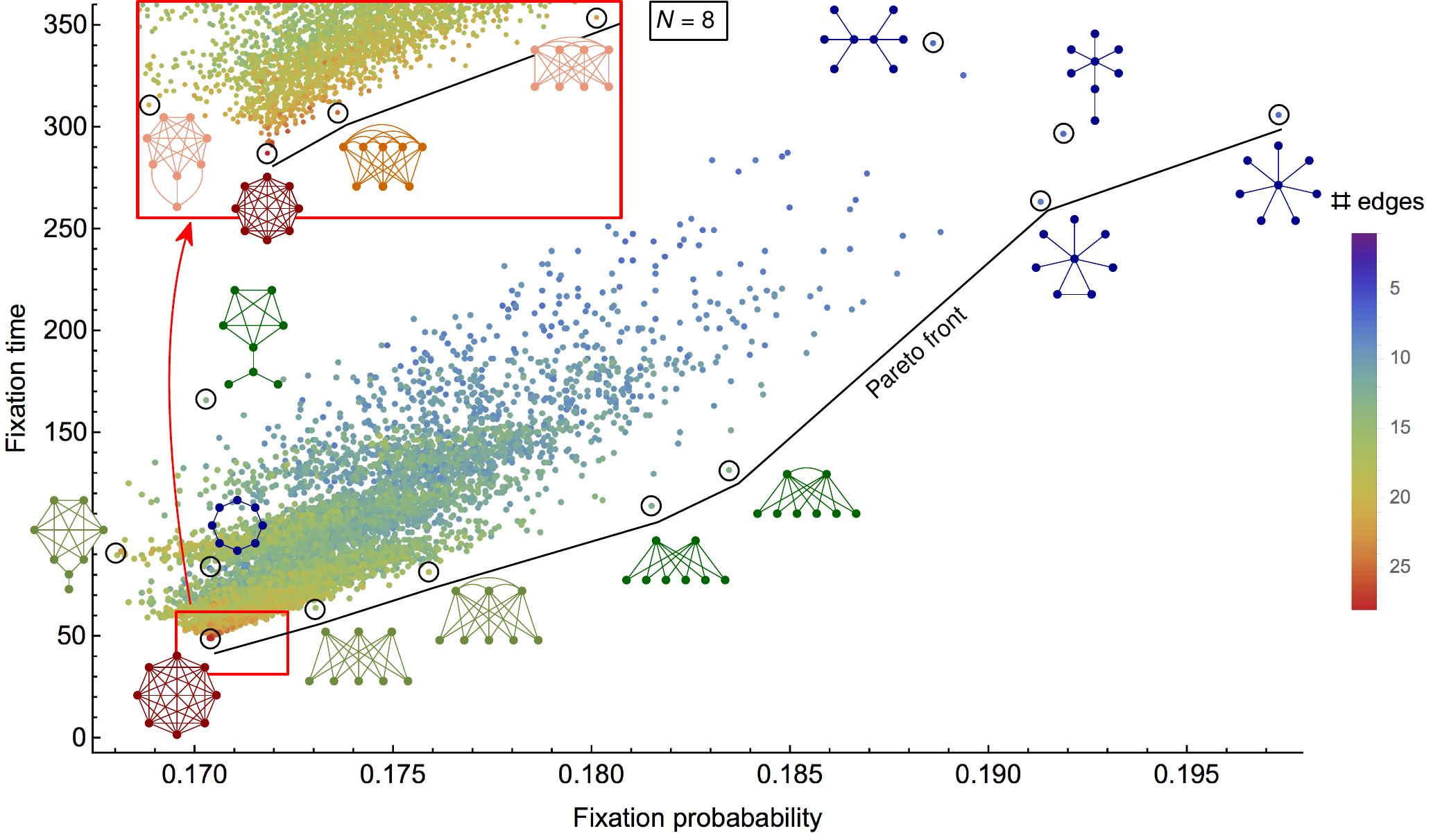}
\caption{
{\bf Fixation probability and time under uniform initialization.}
Numerical solutions for all 11,117 undirected connected graphs of size $N=8$. Each graph is represented by a dot and color corresponds to the number of its edges.
The $x$- and $y$-coordinates show the fixation probability and the fixation time for a single mutant with relative fitness $r=1.1$, under uniform initialization.
The graphs to the right of the complete graph are amplifiers of selection: they increase the fixation probability.
Any graph below the complete graph would be an accelerator of selection: it would decrease the fixation time.
Graphs close to the bottom right corner provide good trade-off between high fixation probability and short fixation time. All the values are computed by numerically solving large systems of linear equations (see e.g.~\cite{hindersin2016exact}). See Figures~\ref{figsupp:sf1} and~\ref{figsupp:sf2} for other $r$ values and $N=9$.
}
\label{fig:2a}

\end{figure}

\item[(2)]
We answer the second question in the affirmative.
The trade-off results (Figure~\ref{fig:2a}) that we study allow us to obtain graphs
which we call \textit{$\eps$-Balanced bipartite graphs}.
Intuitively, they are defined as follows:
We split the vertices into two groups such that one is much smaller than the other, but both are relatively large. Then we connect every two vertices that belong to different groups. We show that, in the limit of large population size, this bipartite graph achieves the same fixation probability as the star graph and that its fixation time asymptotically approaches that of the complete graph.
Formally, an $\eps$-Balanced bipartite graph $B_{N,\eps}$ is a complete bipartite graph with the parts containing $N$ and $N^{1-\eps}$ vertices (see Figure~\ref{fig:2b}a for illustration with $N=8$ and $\eps=1/3$).
We show that the fixation probability of such graphs tends to $1-1/r^2$ while the fixation time is of the order of $N^{1+\eps}\log N$, for any $\eps>0$ 
(compared to $N \log N$ fixation time of complete graph).
Thus we achieve the best of two worlds, that is, we present a graph family that, in the limit of large 
population size, is as good an amplifier as the star graph and almost as good with respect to time as 
the complete graph.
As a byproduct, we prove that on a star graph, both the absorption and the fixation time are of the order of $N^2\log N$ for any fixed $r>1$ which is in alignment with known bounds and approximation results~\cite{broom2011stars,askari2015analytical}.
\end{itemize}
Moreover, we support the analytical result with computer simulations for fixed population size $N=100$. We compute the fixation probability and time for selected families of graphs, such as Trees or random Erd\H os--R\' enyi graphs (see Figure~\ref{fig:2b}b). The $\eps$-Balanced bipartite graphs outperform all of them. Hence the analytical results are interesting not only in the limit of large population but already for relatively small population sizes.

%
\begin{figure}
\includegraphics[width=\linewidth]{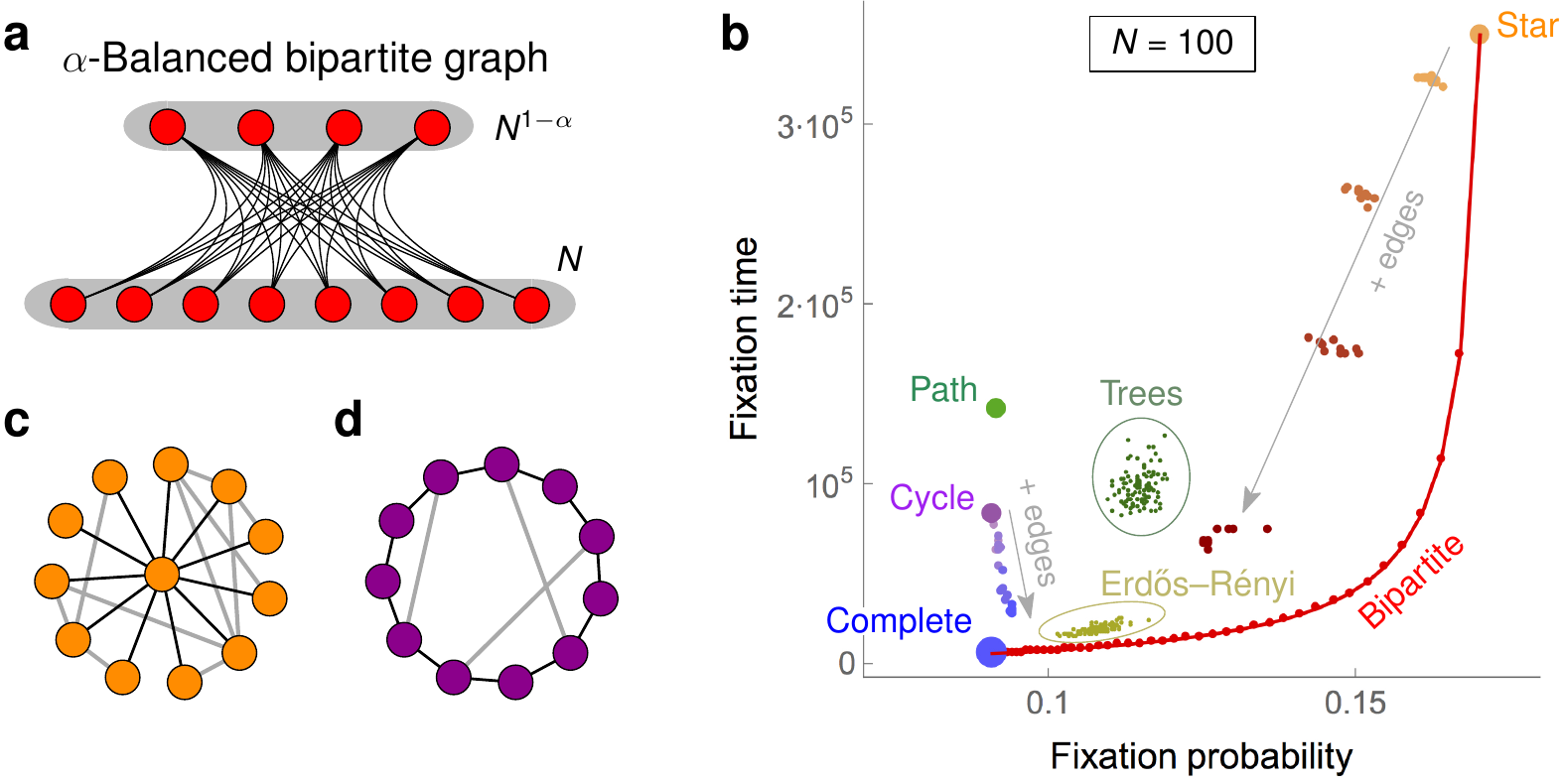}
\caption{
{\bf $\eps$-Balanced bipartite graphs.}
\textbf{a,} An $\eps$-Balanced bipartite graph $B_{N,\eps}$ is a complete bipartite graph with $N$ vertices in the larger part and $N^{1-\eps}$ vertices in the smaller part. Here $N=8$ and $\eps=1/3$. We prove that for large $N$, the $\eps$-Balanced bipartite graphs achieve the fixation probability of a star and, for $\eps$ small, approach the fixation time of the complete graph (see Figures~\ref{figsupp:sf3} and ~\ref{figsupp:sf4}).
\textbf{b,} In general, bipartite graphs provide great trade-offs between high fixation probability and short fixation time.
Comparison is with selected graphs of size $N=100$ such as Trees ($100\times$), random Erd\H os--Renyi graphs ($100\times$, $p=0.03$), star graphs with additional 10, 30, 50, 100 random edges ($10\times$ each), and cycle graphs with additional 1, 3, 5, 10 random edges ($5\times$ each). The values were obtained by simulating the Moran process $10^5$ times.
\textbf{c,} Star graph with several random edges.
\textbf{d,} Cycle graph with several random edges.
}
\label{fig:2b}\end{figure}

\subsection{Temperature initialization: $\eps$-Weighted bipartite graphs}
Third, we consider temperature initialization. 
The above questions for uniform initialization are also the relevant questions for temperature initialization.
Our results are as follows:
\begin{itemize}
\item[(1)] {\em Simulation results.} 
Again we numerically compute the fixation probability and the fixation time for all graphs with $N=8$ vertices (see Figure~\ref{fig:3a}).
In contrast to the results for uniform initialization (Figure~\ref{fig:2a}), under temperature initialization, the complete 
graph has both the highest fixation probability and the shortest fixation time.
This finding holds for other values of $r>1$ and $N$ as well (see Figures~\ref{figsupp:sf5} and~\ref{figsupp:sf6}).

%
\begin{figure}
\includegraphics[width=\linewidth]{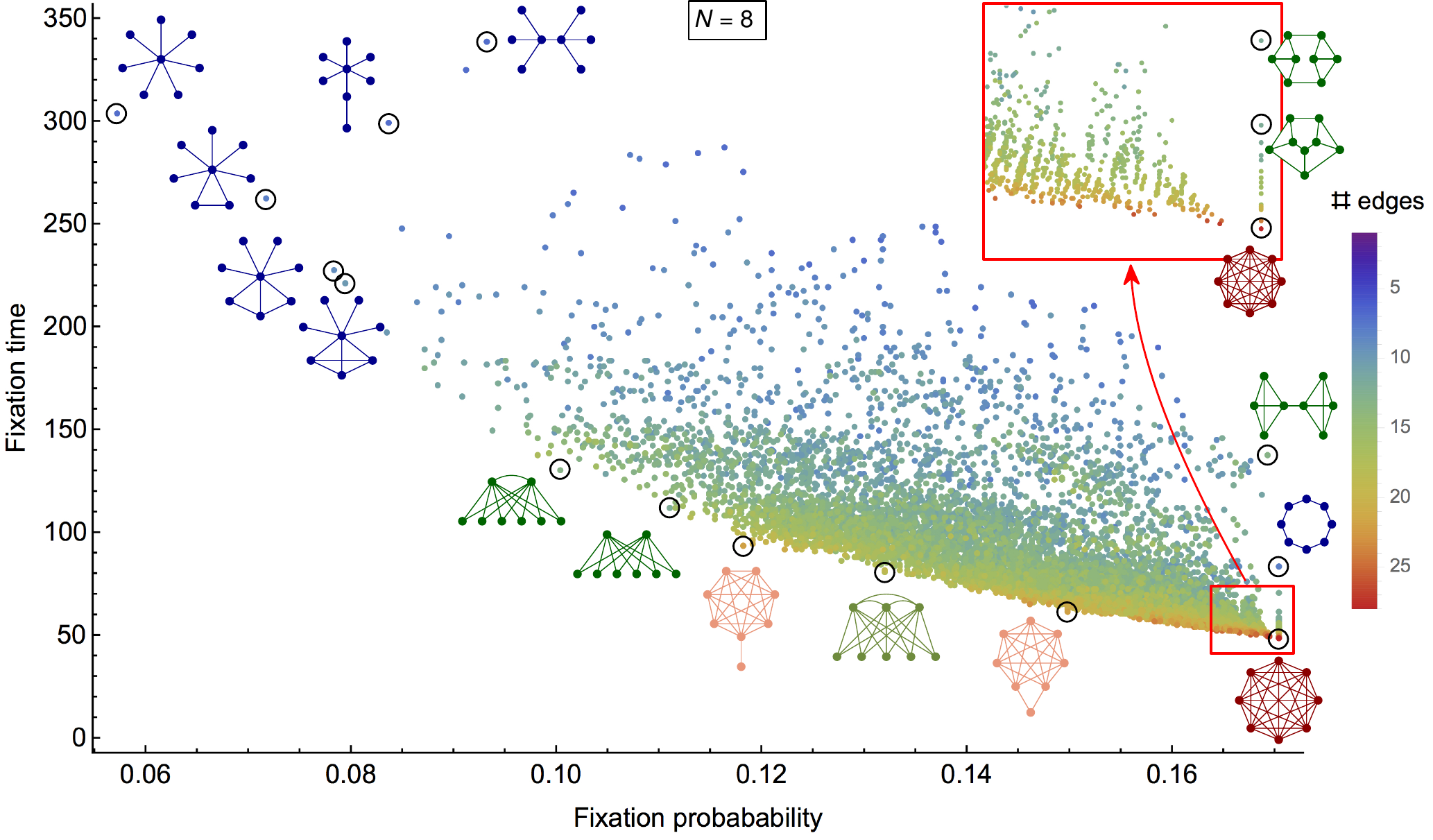}
\caption{
{\bf Fixation probability and time under temperature initialization.}
Numerical solutions for all undirected connected graphs of size $N=8$, under temperature initialization ($r=1.1$). There are no amplifiers and no (strict) accelerators. By the isothermal theorem~\cite{lieberman2005}, all the regular graphs achieve the same fixation probability as the complete graph.
See Figures~\ref{figsupp:sf5} and~\ref{figsupp:sf6} for other $r$ values and $N=9$.
}
\label{fig:3a}

\end{figure}

\item[(2)] {\em Analytical results.} 
Figure~\ref{fig:3a} shows that there is no trade-off for temperature initialization. 
The result is not surprising as it has recently been shown that, for temperature initialization, no unweighted graphs can achieve substantial amplification~\cite{pavlogiannis2018construction},
and in the present work we have established that the complete graph is asymptotically optimal among amplifiers with respect to absorption time.
Thus, the relevant analytical question is whether weighted graphs can achieve interesting trade-offs
between fixation probability and time. 
We answer this question in the affirmative by presenting a weighted version of $\eps$-Balanced bipartite graphs (see Figure~\ref{fig:3b}a).
Intuitively, we add weighted self-loops to all vertices in the larger group of an $\eps$-Balanced bipartite graph, such that when such a vertex is selected for reproduction, its offspring replaces the parent most of the time and migrates to the smaller group only rarely.
 Formally, the $\eps$-Weighted bipartite graph $W_{N,\eps}$ is a complete bipartite graph with the parts containing $N$ and $N^{1-\eps}$ vertices. Moreover, each vertex in the larger part has an extra self-loop of weight approximately $N^{1-\eps/2}$. 
 We show that, in the limit of large population size, this weighted bipartite graph structure achieves 
 fixation probability $1-1/r^2$ (which is the same as the star graph under uniform initialization),
while the fixation time is of the order of $N^{1+\eps}\log N$, for any $\eps>0$ 
(compared to $N \log N$ fixation time of complete graph).
Thus we again achieve the best of two worlds, that is, we present a graph family that, in the limit of large 
population, is as good an amplifier as the star graph (under uniform initialization)
and almost as good with respect to time as the complete graph.
\end{itemize}
As before, Figure~\ref{fig:3b}b shows computer simulations for $N=100$, including Trees, random Erd\H os--R\' enyi graphs and the Bipartite graphs. The $\eps$-Weighted bipartite graphs are the only graphs that considerably increase the fixation probability as compared to the complete graph.

%
\begin{figure}
\includegraphics[width=\linewidth]{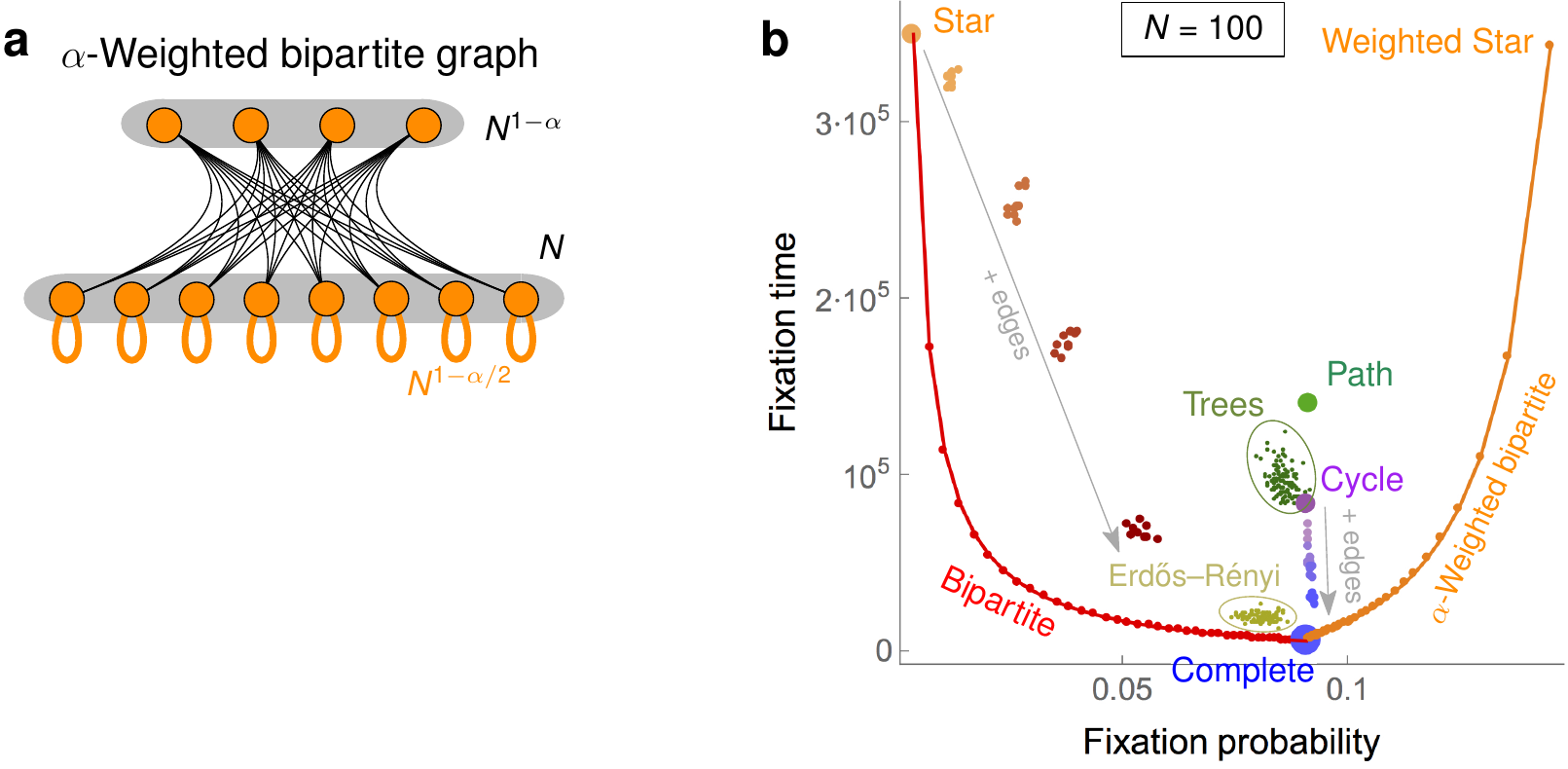}
\caption{
{\bf $\eps$-Weighted bipartite graphs.}
\textbf{a,} An $\eps$-Weighted bipartite graph $W_{N,\eps}$ is obtained by adding self-loops with weight $w\doteq N^{1-\eps/2}$ to all vertices in the larger part of an $\eps$-Balanced bipartite graph. Here $N=8$ and $\eps=1/3$. We prove that for large $N$, the $\eps$-Weighted bipartite graphs improve the fixation probability to $1-1/r^2$ and, for $\eps$ small, approach the fixation time of the complete graph.
\textbf{b,} Computer simulations for selected graphs of size $N=100$ (as in Figure~\ref{fig:2b}b). It is known than among unweighted graphs, only a very limited amplification can be achieved~\cite{pavlogiannis2018construction}. Our $\eps$-Weighted bipartite graphs (with self-loops of varying weight) overcome this limitation and provide trade-offs between high fixation probability and short fixation time.
}
\label{fig:3b}
\end{figure}

\subsection{Effective rate of evolution}
Finally, we study the effectiveness of the presented population structures for small population sizes. 
We use an elegant mathematical formula for the effective rate of evolution that combines both fixation probability 
and fixation time~\cite{Frean2013}. 
Let $t_1= \frac{1}{N \mu \rho}$ denote the expected number of generations to generate a mutant that eventually fixates,
where $N$ is the population size, $\mu$ is the mutation rate and $\rho$ is the fixation probability.
Let $t_2= \frac{\tau}{N}$ denote the expected number of generations for a mutant to fixate once it is generated.
Note that $\tau$ is the fixation time measured in steps of the Moran process, and $\frac{\tau}{N}$ represents the number of generations.
The effective rate of evolution is defined as the inverse of the sum of the above two quantities, i.e., $\frac{1}{t_1+t_2}$.
The effective rate of evolution was studied for the complete graph and for the star graph under uniform initialization~\cite{Frean2013}.
Here we further investigate the effective rate of evolution for $\eps$-Balanced bipartite graph under uniform initialization, and for $\eps$-Weighted bipartite graphs under temperature initialization, for relatively small population sizes.

Regarding uniform initialization, we numerically compute the effective rate of evolution on $\eps$-Balanced bipartite graphs for a wide range of mutation rates $\mu$ and compare it to the effective rate of evolution on complete graphs and star graphs (see Figure~\ref{fig:4}a for fixed population size and Figure~\ref{fig:4}b for varying population sizes). 
The complete graph is more effective for high mutation rates and the star graph is more effective for low mutation rates but in the intermediate regime, suitable $\eps$-Balanced bipartite graphs are more effective than both the complete graph and the star graph. This is in a perfect alignment with the Pareto front presented in Figure~\ref{fig:2a}.

Regarding temperature initialization, we study $\eps$-Weighted bipartite graphs instead of $\eps$-Balanced bipartite graphs (Figure~\ref{fig:4}c,d).
As before, the complete graph is the most effective population structure for high mutation rates. However, star graph is a suppressor under temperature initialization and performs poorly. Therefore, except for the high mutation rate regime, various $\eps$-Weighted bipartite graphs achieve higher effective rate of evolution than both the complete graph and the star graph.

%
\begin{figure}
\includegraphics[width=\linewidth]{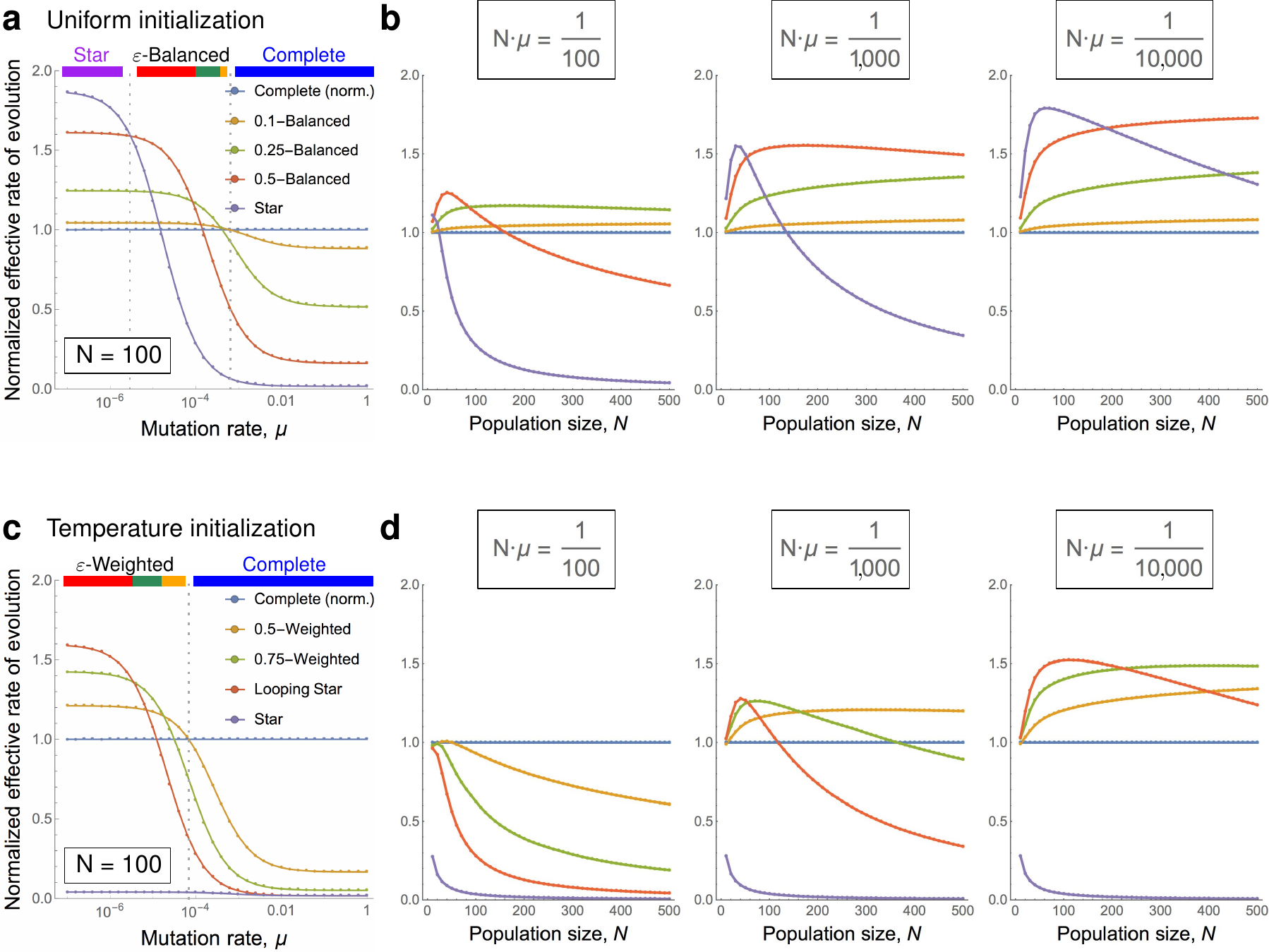}
\caption{
{\bf Fig.~4. Effective rate of evolution.}
The effective rate of evolution depends on the population size, $N$, the mutation rate, $\mu$, and the population structure.
For uniform initialization, we compare five different population structures: the complete graph (blue), $\eps$-Balanced graphs with $\eps\in\{0.1,0.25,0.5\}$ (orange, green, red), and the star graph (purple), always showing the relative rate of evolution with respect to the complete graph.
\textbf{a,} We fix $N=100$, $r=1.1$ and vary $\mu=10^{-7},\dots,10^{0}$. The complete graph has a higher effective rate of evolution if the mutation rate is high ($\mu>10^{-3}$) and star graph is favorable if the mutation rate is low ($\mu<3\cdot 10^{-6}$). In the intermediate regime, suitable $\eps$-Balanced graphs outperform both of them.
\textbf{b,} We fix $r=1.1$ and $N\cdot \mu\in\{10^{-2},10^{-3},10^{-4}\}$ and vary $N=10,20,\dots,500$.
The star graph is favorable if mutations are rare ($N\cdot\mu=10^{-4}$ and $N$ small).
Otherwise, suitable $\eps$-Balanced graphs are more efficient.
%
\textbf{c, d} Analogous data for temperature initialization. This time we compare the complete graph (blue) and the star graph (purple) with $\eps$-Weighted bipartite graphs for $\eps\in\{0.25,0.5,1\}$ (orange, green, red). 
The complete graph dominates if mutations are common ($N\cdot\mu=10^{-2}$). In other cases, $\eps$-Weighted bipartite graphs are preferred. The star is not an amplifier for temperature initialization.
}
\label{fig:4}
\end{figure}


\section{Discussion}

Many previous studies have explored how population structure affects the fixation probability of new mutants~\cite{lieberman2005,Houchmandzadeh2011,Broom2008,Broom2011,shakarian12,Allen2017,Galanis17,Giakkoupis16,Goldberg16,pavlogiannis2018construction,Diaz2013,jamieson2015fixation}. 
While such studies cover one major aspect of evolutionary dynamics, the other aspect, which is fixation time, is much less studied.
Both fixation probability and fixation time play an important role in determining the rate of evolution.
If the mutation rate is low, the rate-limiting step is waiting for an advantageous mutant to occur. In this regime the fixation probability is more important than the fixation time. Conversely, if the mutation rate is high, then fixation time is more relevant than fixation probability. In the intermediate-mutation rate regime, the trade-off between fixation probability and fixation time must be considered.
We study this trade-off and propose population structures, called $\eps$-Balanced bipartite graphs and
$\eps$-Weighted bipartite graphs,
that provide substantial amplification with negligible increase in the fixation time.
This is in stark contrast with all previous works that achieve amplification at the cost of asymptotically increasing the fixation time. As a consequence, compared to previous works, our population structures enable higher effective rate of evolution than the well-mixed population for a wide range of mutation-rate regimes.

There are some interesting mathematical questions that remain open. 
While we show that (i) amplifiers cannot have better asymptotic absorption time than the well-mixed population (in the limit of large population size, $N\to \infty$), and (ii) there are graphs of fixed population size $N$, that are suppressors and have shorter fixation time than the well-mixed population, 
two particularly interesting questions are:
(a)~Does there exist an amplifier of fixed population size that has shorter fixation time than the well-mixed population?
(b)~Does there exist a graph family (which must be suppressing) that has better asymptotic fixation time (for $N\to \infty$) than the well-mixed population?

Note that, in general, clonal interference can occur and the fixation of a mutant need not be achieved
before the next mutation arrives~\cite{Desai07,gerrish1998fate,fogle2008clonal}.
Thus, the fixation probability and fixation time alone may not completely characterize the performance 
of a population structure with respect to the overall rate and  efficiency of an evolutionary search process.
Nevertheless, the effective rate of evolution and the probability-time trade-off curves are indicative of the efficacy of each population structure in speeding-up evolution. The numerical and experimental study of population structures in the presence of clonal interference is another interesting direction for future work.

The population structures which we have described here could become an important tool for in vitro evolution~\cite{Mitchell09,Barrick09,Dai12,Lang13}, 
since they can substantially speed up the process of finding advantageous mutants.
In vitro evolution, can be used to discover optimized protein or nucleotide sequences for any medical or industrial purpose. 
Depending on the mutation-rate regime, our work shows that different population structures can lead to more effective time scales of discovery.

%
%
\section{Methods}

In this section we introduce the model in detail and formally state our results, pointing to the relevant appendices for the full derivation.  

\subsection*{Moran process on graphs}
Moran Birth-death process is a discrete-time stochastic (random) process that models evolutionary dynamics in a spatially structured population.
The population structure is represented by a connected graph $G$, possibly with weighted edges and/or self-loops.
At all times, each vertex of the graph is occupied by a single individual that is of one of two types: either a \textit{resident} or a \textit{mutant}.
The individuals of one type are considered indistinguishable.
Moreover, residents are assigned (normalized) fitness 1 while the mutants have fitness $r$. Here we consider advantageous mutants ($r>1$).
In one step of the process, an individual is selected for reproduction randomly and proportionally to its fitness. This individual produces an offspring that is a copy of itself. This offspring then selects one of the adjacent edges proportionally to the edge weight and travels along that edge to replace the individual at its other endpoint. (If the selected edge happened to be a self-loop then the offspring replaces the parent and nothing changes.)
These steps continue until the population becomes homogeneous: either all individuals are mutants (\textit{fixation} occurred) or they are all residents (\textit{extinction} occurred).
The well-mixed population is modelled by an unweighted complete graph (without self-loops).

\sne{Initialization scheme} We study the situation of a single mutant invading a population of residents. This initial mutant can appear either spontaneously or during reproduction. In the first case, called \textit{uniform initialization}, the mutant is placed at a vertex chosen uniformly at random. In the second case, called \textit{temperature initialization}, we perform one step of the Moran process in a population that consists entirely of residents and place the mutant at the vertex that the offspring migrates to.
Formally, the mutant is placed at a random vertex, proportionally to the temperature (or turnover rate) of that vertex. Here temperature $t(v)$ of a vertex $v$ is defined by
$$t(v)=\sum_{u\in N(v)} \frac{w(u,v)}{\sum_{v'\in N(u)} w(u,v')},
$$
where $w(u,v)$ is the weight of edge between $u$ and $v$ and $N(v)$ is the set of \textit{neighbors} of $v$, that is vertices connected to $v$ by an edge.

\sne{Fixation probability and time} 
Given a graph $G$ with $N$ vertices and one specific vertex $v$, we denote by $\fp(G,v,r)$ the fixation probability of a single mutant with fitness $r$ starting at vertex $v$, in a standard Moran Birth-death process. 
Then the fixation probability under uniform initialization is simply the average
$\fp(G,r)=\frac1N\sum_v \fp(G,v,r)$.
The fixation probability under temperature initialization is a weighted average
$\fp_\temp(G,r)=\sum_v t(v)\cdot \fp(G,v,r)$, where $t(v)$ is the temperature of vertex $v$.
Similarly, we define $T(G,r)$ (or $T_\temp (G,r)$) to be the fixation time, that is the expected number of steps of the Moran process until the mutants reach fixation (conditioning on them doing so). Likewise we define $\ET(G,r)$ (or $\ET_\temp (G,r)$) to be the extinction time and $\AT(G,r)$ (or $\AT_\temp (G,r)$) to be the (unconditional) absorption time.

\sne{Amplifiers and suppressors} A graph $G_N$ with $N$ vertices is called an \textit{amplifier} if it increases the fixation probability of any advantageous mutant, as compared to the Complete graph (that is, $\fp(G_N,r)>\fp(K_N,r)$ for any $r>1$).
On the other hand, a graph $G'_N$ with $N$ vertices is called a \textit{suppressor} if it decreases the fixation probability of any advantageous mutant, as compared to the Complete graph (that is, $\fp(G'_N,r)<\fp(K_N,r)$ for any $r>1$).

\subsection*{Notation for asymptotic behavior}
To talk about asymptotic behavior (in the limit of large population size $N$), we use standard mathematical notations $o(\cdot)$, $O(\cdot)$, and $\Theta(\cdot)$ that denote asymptotically strictly smaller, asymptotically less than or equal to, and asymptotically equal to (up to a constant factor), respectively. 
For example, we will write $1/N=o(1)$ (as $1/N$ is much smaller than 1, for large $N$) or $\frac12N(N+1)=\Theta(N^2)$. For detailed treatment see~\cite[Section~1.3]{cormen2009introduction}.

\subsection*{Graphs}\label{subsec:graphs}
We introduce and study the following graphs.

\sne{Complete graph} Complete graph $K_N$ on $N$ vertices models a well-mixed population. This case is well understood. In particular, the fixation probability satisfies
$$\fp(K_N,r)=\fp_\temp(K_N,r)=\frac{1-1/r}{1-1/r^N}\to 1-1/r
$$
for $r>1$ as $N\to\infty$ and 
the (unconditional) absorption time is of the order of $\Theta(N\log N)$~\cite{diaz2016absorption}.
In fact, using a standard difference method one can derive that, for $r>1$, we have $\AT(K_N,r)\approx\frac{r+1}{r}\cdot N\log N$ and $T(K_N,r)\approx\frac{r+1}{r-1}\cdot N\log N$. For reference purposes we present those proofs in Appendix 4.

\sne{Star graph} Star graph $S_N$ consists of one central vertex connected to each of the remaining $N-1$ vertices on the periphery. For large $N$, it is known that $\fp(S_N,r)\to 1-1/r^2$
and that the absorption and fixation time are of the order of at most $O(N^2\log N)$ and $O(N^3)$, respectively~\cite{broom2011stars}.
 In fact, as a corollary of our results on $\eps$-Balanced bipartite graph, we show that both the absorption time and the fixation time are of the order of $\Theta(N^2\log N)$. 
The bottom line is that, under uniform initialization, the Star graph amplifies the fixation probability but at the cost of substantially increasing the fixation time.
Under temperature initialization, the star graph is a suppressor (in fact, $\fp_\temp(S_n,r)\to 0$ as $n\to\infty$). 

\sne{$\eps$-Balanced bipartite graph}
For uniform initialization we present a family of graphs that, in the limit of large population size, achieve the fixation probability of the Star graph and the fixation time almost as good as the Complete graph. The graphs are complete bipartite graphs with both parts large but one part asymptotically larger than the other one.
Formally, given $N$ and $\eps\in(0,1]$, the $\eps$-Balanced bipartite graph $B_{N,\eps}$ is a complete bipartite graph with parts of size $N^{1-\eps}$ and $N$. That is, there are $N^{1-\eps}$ vertices in one part, $N$ vertices in the other part, and all edges that connect vertices in different parts. The case $\eps=1$ corresponds to a Star graph.

\sne{$\eps$-Weighted bipartite graphs}
For temperature initialization, the Star graph and the $\eps$-Balanced bipartite graphs fail to amplify. We present another family of weighted graphs with self-loops that, in the limit of large population size, provide fixation probability $1-1/r^2$ (the same as Star graph under uniform initialization) and the fixation time almost as good as the Complete graph.
The graphs are obtained by adding self-loops of relatively large weight to all vertices in the larger part of an $\eps$-Balanced bipartite graph.
Formally, given $N$ and $\eps\in(0,1)$, the Weighted bipartite graph $W_{N,\eps}$ is a complete bipartite graphs with one (smaller) part of size $N^{1-\eps}$, one (larger) part of size $N$, and every vertex of the larger part having a self-loop of such a weight $w$ that $N^{-\eps/2}=\frac{N^{1-\eps}}{w+N^{1-\eps}}$.
The case $\eps=1$ is closely related to a Looping Star~\cite{ACN15}.

\subsection*{Analytical results}
Here we summarize our analytical results.
They are all related to the trade-off between fixation probability and fixation time, under both uniform and temperature initialization.

First, we prove that no amplifier is asymptotically faster than the Complete graph in terms of absorption time (recall that $T(K_N,r)=\Theta(N\log N)$, see Appendix~4). 
Informally, the idea is as follows: For every $k=1,\dots,N-1$, we denote by $t_k^X$ the expected time it takes to gain a single mutant from any configuration $X$ consisting of $k$ mutants. To gain a mutant, one of the $k$ mutants has to be selected for reproduction (and then the mutant has to replace a resident). We show that this yields a lower bound for $t_k$ that is proportional to $N/k$.
Summing over all $k$'s we get that the total absorption time is of the order of at least
$N/1+N/2+\dots+N/(N-1)\approx N \log N$.
Since the absorption time for the complete graph is also proportional to $N\log N$, no amplifier is significantly faster than the complete graph.

\begin{theorem}\label{thm:noaccelerators}
Fix $r>1$. Let $G$ be any graph with $N\ge 2$ vertices and let $p=\fp(G,r)$ be the fixation probability of a single mutant under uniform initialization. Then
$$\AT(G,r)\ge \frac pr\cdot N\cdot H_{N-1},$$
where $H_{N-1}=\frac11+\frac12+\dots+\frac1{N-1}\ge \log N$.
In particular, $\AT(G,r)\ge \frac pr\cdot N\log N$ for an arbitrary graph $G$
and $\AT(A,r)\ge \frac{r-1}{r^2}\cdot N\log N $ for an arbitrary amplifier $A$.
\end{theorem}
For the formal proof, see Appendix 1.

Second, we give tight results for the fixation time on Bipartite graphs. In particular, we prove that under uniform initialization, certain $\eps$-Balanced bipartite graphs $B_{N,\eps}$ asymptotically achieve the fixation probability of the Star graph and the fixation time almost as good as the Complete graph.
The analysis of fixation probability is relatively straightforward. 
For fixation time, we provide tight lower and upper bounds. We first present the lower bound that is proportional to $N^{1+\eps}\log N$. For the upper bound we then distinguish two cases: If the size of the smaller part is small, that is $N^{1-\eps}=o(\sqrt N)$, then the argument is simpler and we get a matching upper bound. If the size of the smaller part is relatively close to $N$, the upper bound has an additional factor of $N^\eps$. As a consequence, we can prove the following theorem.

\begin{theorem}\label{biptime} Fix $\eps\in(0,1]$ and $r>1$. Let $B_{N,\eps}$ be the $\eps$-Balanced bipartite graph. Then
\begin{itemize}
\item $\fp(B_{N,\eps},r)\to 1-1/r^2$.
\begin{itemize}
\item (small center) If $\eps\in(0.5,1)$ then there exist constants $c_1$, $c_2$ such that 
$$c_1\cdot N^{1+\eps}\log N \le \AT(B_{N,\eps},r) \le c_2\cdot N^{1+\eps}\log N.$$
\item (large center) If $\eps\in(0,0.5)$ then there exist constants $c_1$, $c_2$ such that 
$$c_1\cdot N^{1+\eps}\log N \le \AT(B_{N,\eps},r) \le c_2\cdot N^{1+2\eps}\log N.$$
\end{itemize}
\end{itemize}
Moreover, the fixation time $ T(B_{N,\eps},r)$ satisfies the same inequalities.
\end{theorem}
As an immediate corollary, we obtain that for any fixed $r>1$, both the absorption and the fixation time on a Star graph ($\alpha=1$) are of the order of $\Theta(N^2\log N)$. This is in alignment with earlier results~\cite{broom2011stars,askari2015analytical}.
For the formal proof, see Appendix 2.

Third, we prove that under temperature initialization, analogous results can be achieved using $\eps$-Weighted bipartite graphs $W_{N,\eps}$.
\begin{theorem}\label{weighted} Fix $\eps\in(0,1]$ and $r>1$. Let $W_{N,\eps}$ be the Weighted bipartite graph. Then
\begin{itemize}
\item $\fp(W_{N,\eps},r)\to 1-1/r^2$.
\item There exist constants $c_1$, $c_2$ such that 
$$c_1\cdot N^{1+\eps}\log N \le \AT(B_{N,\eps},r) \le c_2\cdot N^{1+\frac32\eps}\log N.$$
\end{itemize}
Moreover, the fixation time $ T(B_{N,\eps},r)$ satisfies the same inequalities.
\end{theorem}
For the formal proof, see Appendix 3.

Finally, for reference purposes we compute the absorption, fixation, and extinction times of a single advantageous mutant ($r>1$) on a Complete graph, using the standard difference method.

\begin{theorem} Fix $r>1$ and let $K_N$ be the Complete graph on $N$ vertices. Then
\begin{align*}
\AT(K_N,r)&=(N-1)H_{N-1}\cdot \frac{r+1}r + (N-1)\cdot\log(1-1/r) - \frac1{r(r-1)} +o(1),\\
T(K_N,r)&=(N-1)H_{N-1}\cdot \frac{r+1}{r-1} + (N-1)\cdot\frac{r+1}{r-1}\log(1-1/r)+o(N),\\
\ET(K_N,r)&=  (N-1)\cdot \log\left(\frac r{r-1}\right) +o(N).
\end{align*}
In particular for $r=1+s$, $s>0$ small, we have  $\AT(K_N,r)\approx 2\cdot N\log N$, $T(K_N,r)\approx\frac 2s\cdot N\log N$, and $\ET(K_N,r)\approx \frac1s\cdot N $.
\end{theorem}
For the formal proof, see Appendix 4.
\setcounter{theorem}{0}

\section*{Acknowledgments}
J.T. and K.C. acknowledge support from ERC Start grant no. (279307: Graph Games), Austrian Science Fund (FWF) grant no. P23499-N23 and S11407-N23 (RiSE). 
A.P. acknowledges support from FWF Grant No. J-4220.
M.A.N. acknowledges support from Office of Naval Research grant N00014-16-1-2914 and from the John Templeton Foundation. 
The Program for Evolutionary Dynamics is supported in part by a gift from B. Wu and E. Larson.

\section*{Competing interests} The authors declare that no competing interests exist.

\bibliography{refs}


\pagebreak
\section{Additional Figures}

%
\begin{figure}[h]
\center
\includegraphics[width=0.8\linewidth]{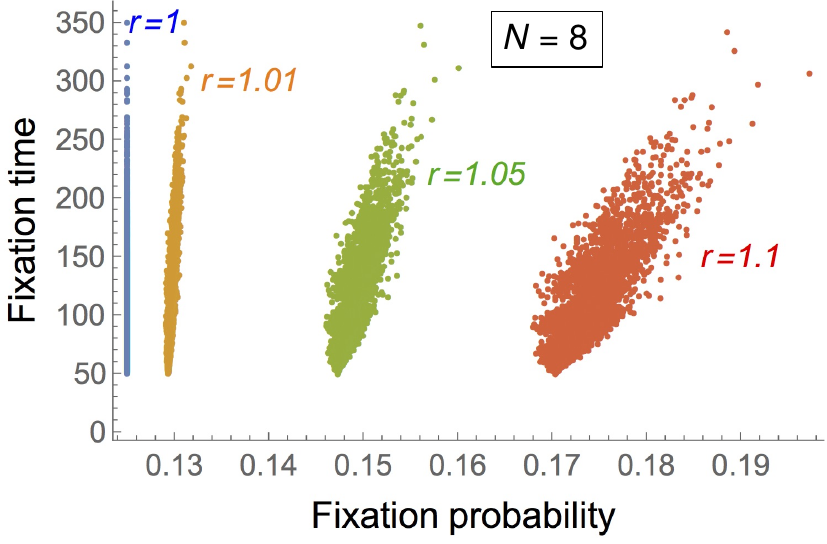}
\caption{
\textbf{Fixation probability and time under uniform initialization, other $r$ values.}
Similar data as in Figure~\ref{fig:2a} for varying $r\in\{1,1.01,1.05,1.1\}$. Under uniform initialization, the fixation probability of a neutral mutant equals $1/N$, independent of the graph structure. As $r$ approaches 1, the point cloud gets closer to a vertical line.
}
\label{figsupp:sf1}
\end{figure}

%
\begin{figure}[h]
\center
\includegraphics[width=\linewidth]{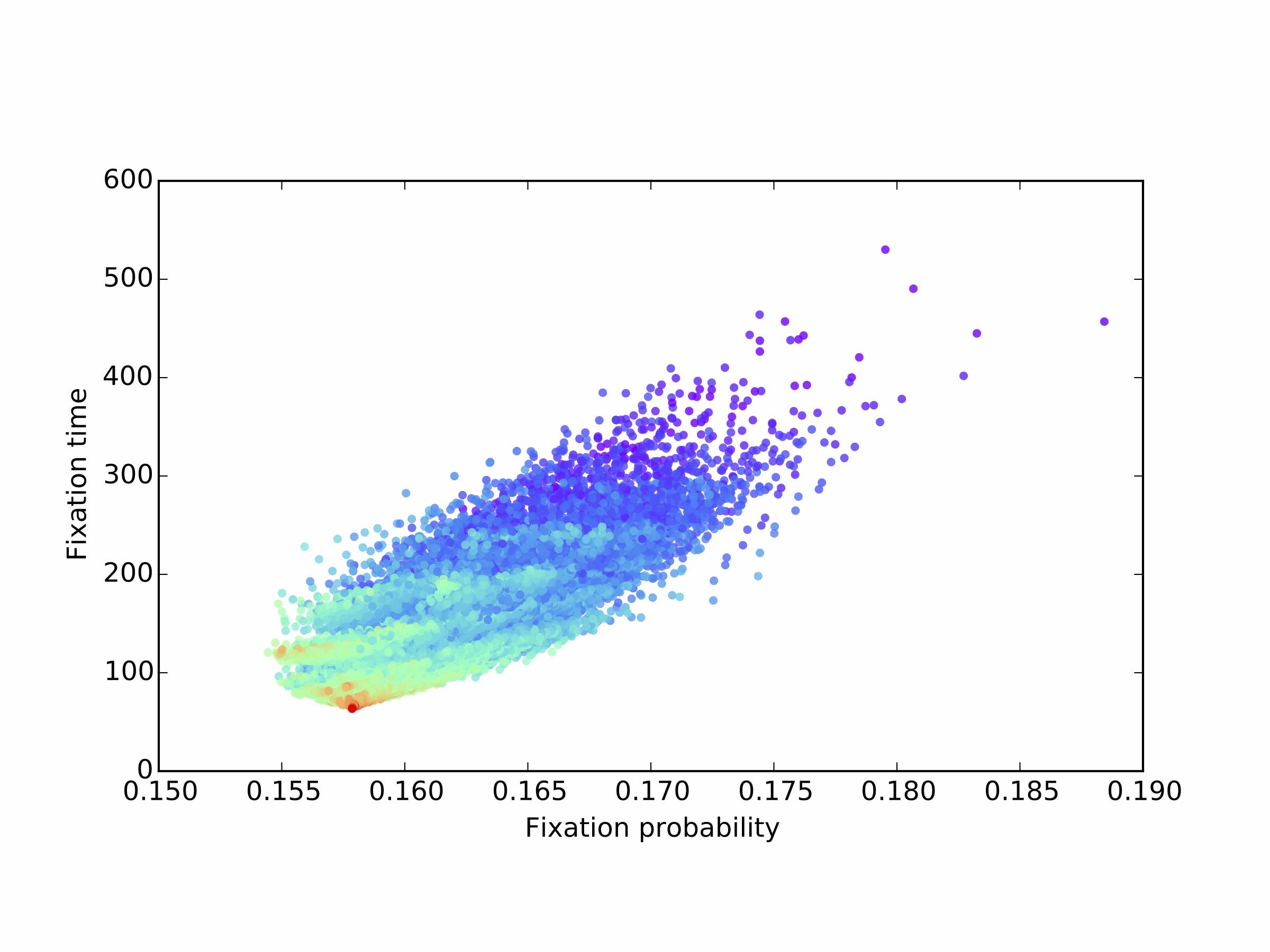}
\caption{
\textbf{Fixation probability and time under uniform initialization, $N=9$.}
Similar data as in Figure~\ref{fig:2a} for all 261,080 graphs of size $N=9$ (here $r=1.1$). The results are qualitatively the same as for $N=8$.
}
\label{figsupp:sf2}
\end{figure}

\clearpage
%
\begin{figure}[h]
\center
\includegraphics[width=0.8\linewidth]{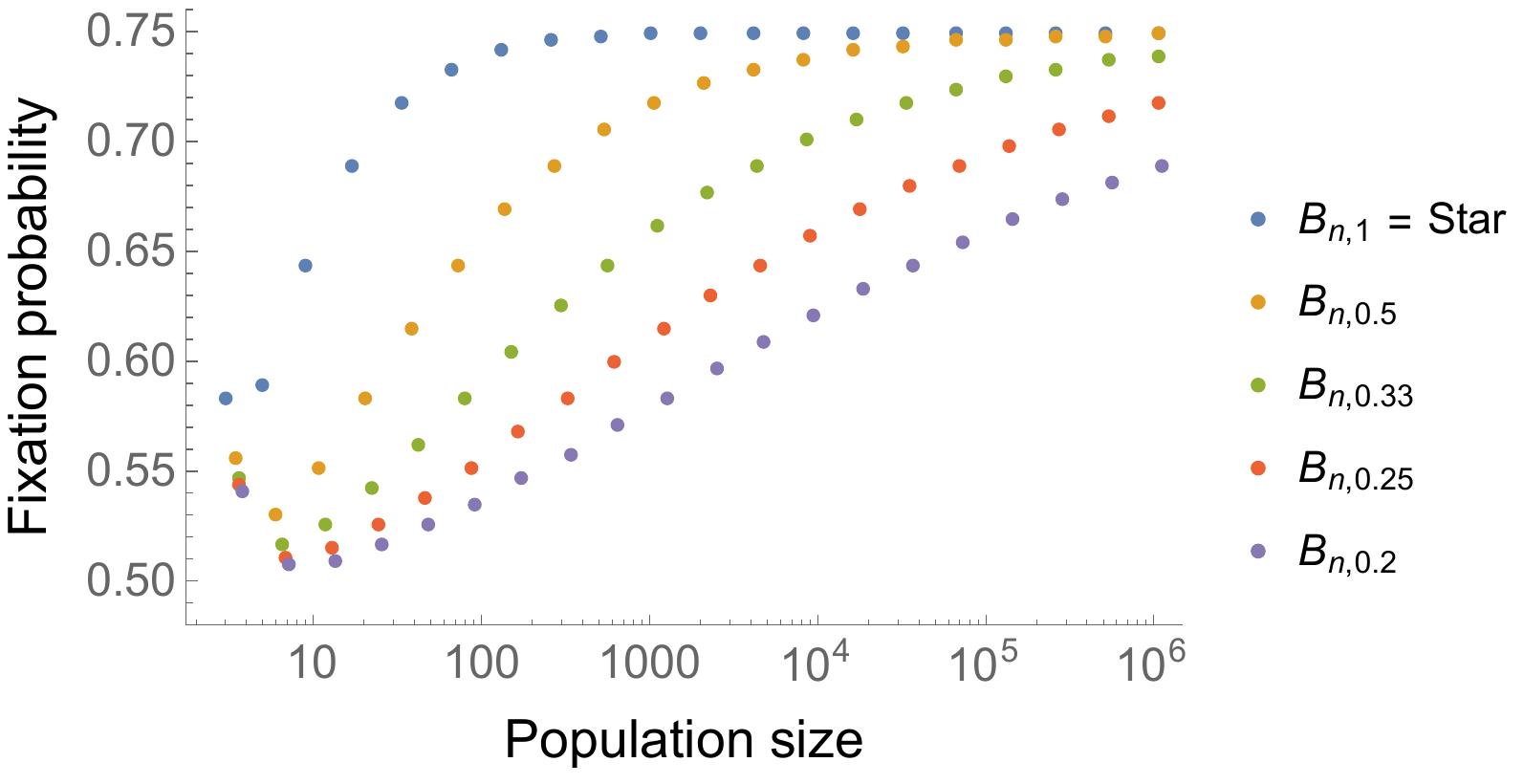}
\caption{
\textbf{Fixation probability for $B_{n,\eps}$ tends to $1-1/r^2$.}
We fix $r=2$ and consider the $\eps$-Balanced bipartite graphs $B_{n,\eps}$ for $\eps=1$ (i.e. a star) and $\eps\in\{0.5,0.33,0.25,0.2\}$. The dots are exact values of the fixation probability under uniform initialization. The figure illustrates that the fixation probability tends to $1-1/r^2$ for any $\eps>0$.
}
\label{figsupp:sf3}
\end{figure}

\pagebreak
%
\begin{figure}[h]
\center
\includegraphics[width=\linewidth]{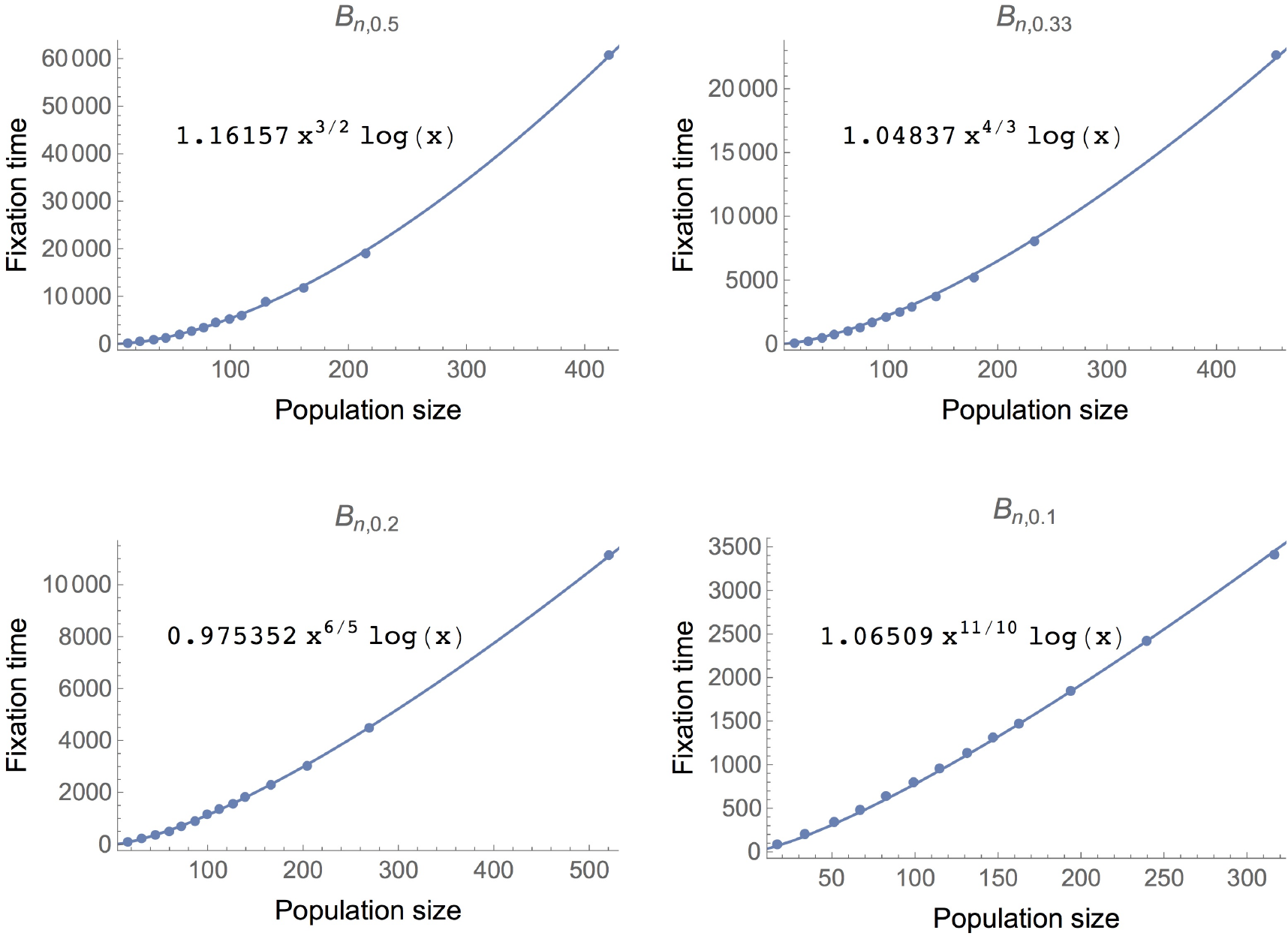}
\caption{
\textbf{Fixation time for $B_{n,\eps}$ is proportional to $n^{1+\eps}\log n$.}
We fix $r=2$ and consider the $\eps$-Balanced bipartite graphs $B_{n,\eps}$ for $\eps\in\{0.5,0.33,0.2,0.1\}$ and for $n$ up to 500. The dots are exact numerical solutions, the lines are the best fits. The figure confirms that the fixation time $T(B_{n,\eps},r)$ is proportional to $n^{1+\eps}\log n$ for any $\eps>0$.
}
\label{figsupp:sf4}
\end{figure}

\pagebreak
%
\begin{figure}[h]
\center
\includegraphics[width=0.8\linewidth]{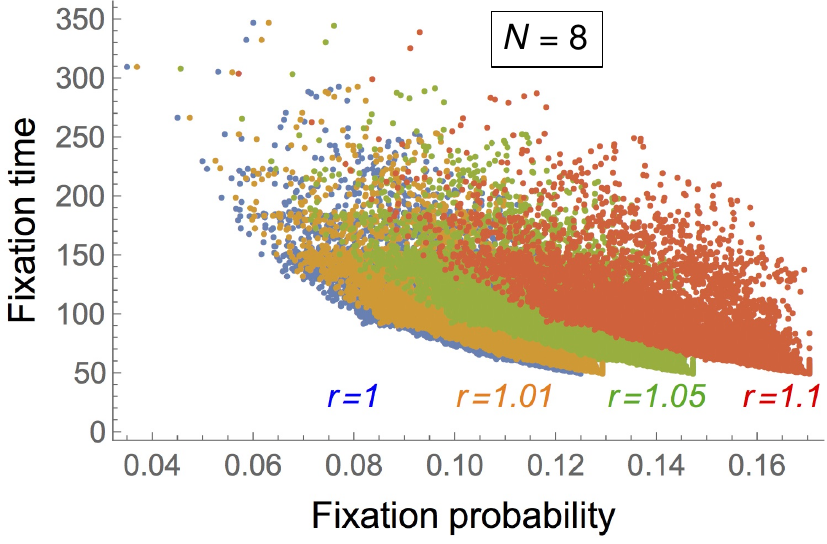}
\caption{
\textbf{Fixation probability and time under temperature initialization, other $r$ values.}
Similar data as in Figure~\ref{fig:3a} for varying $r\in\{1,1.01,1.05,1.1\}$. The results are qualitatively the same as for $r=1.1$.
}
\label{figsupp:sf5}
\end{figure}

\pagebreak
%
\begin{figure}[h]
\center
\includegraphics[width=\linewidth]{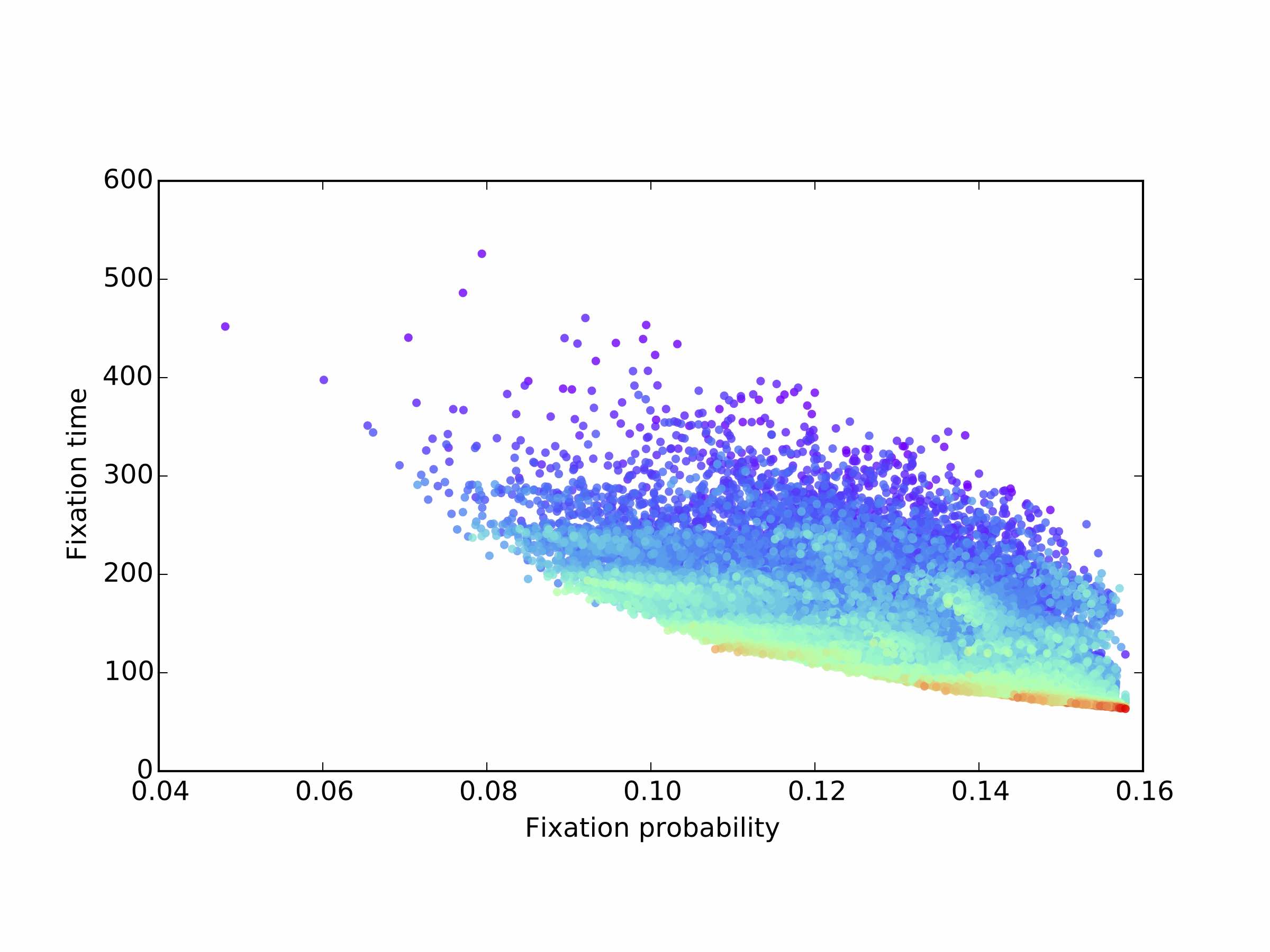}
\caption{
\textbf{Fixation probability and time under temperature initialization, $N=9$.}
Similar data as in Figure~\ref{fig:3a} for all 261,080 graphs of size $N=9$ (here $r=1.1$). The results are qualitatively the same as for $N=8$.
}
\label{figsupp:sf6}
\end{figure}

\pagebreak
%
\begin{figure}[h]
\center
\includegraphics[width=\linewidth]{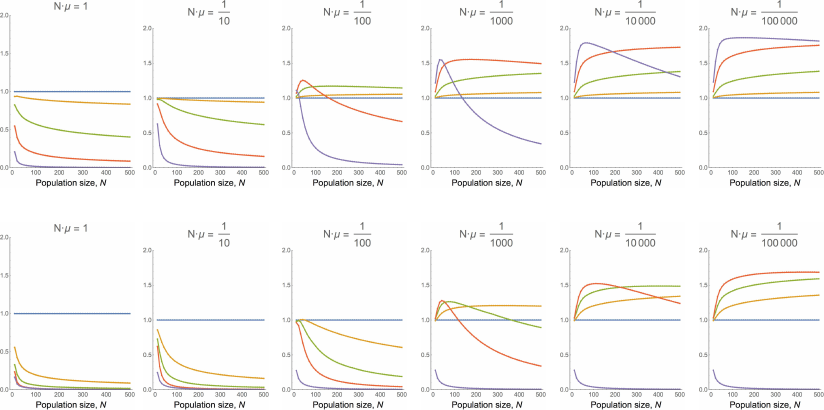}
\caption{
\textbf{Effective rate of evolution, other regimes.}
As in the Figure~\ref{fig:4} from the main text, we fix $r=1.1$ and vary population size $N=10,20,\dots, 500$.
In the first row, we consider uniform initialization and the following graphs:
Complete graph (blue), $\eps$-Balanced graph for $\eps\in\{0.1,0.25,0.5\}$ (orange, green, red), and the Star graph (purple) which is the same as a $1$-Balanced graph.
In the second row, we consider temperature initialization and the following graphs:
Complete graph (blue), Star (purple), and $\eps$-Weighted bipartite graphs for $\eps\in\{0.25,0.5\}$ (orange, green) and Weighted Star ($\eps=1$, red).
}
\label{figsupp:sf7}
\end{figure}


\clearpage
\section{Appendices}
\setcounter{section}{0}
%
\section*{Appendix 1: Lower bound on absorption time}\label{app:1}
Here we show that for $r>1$ no family of graphs with fixation probability bounded away from zero can have asymptotically smaller absorption time than the Complete graphs. Specifically, no amplifiers can absorb asymptotically faster than the Complete graphs.
Recall that for the Complete graph on $N$ vertices, both the fixation time and the absorption time is of the order of $\Theta(N\log N)$ (see Appendix~4).

\begin{theorem}\label{thm:noaccelerators}
Fix $r>1$. Let $G$ be any graph with $N\ge 2$ vertices and let $p=\fp(G,r)$ be the fixation probability of a single mutant under uniform initialization. Then
$$\AT(G,r)\ge \frac pr\cdot N\cdot H_{N-1},$$
where $H_{N-1}=\frac11+\frac12+\dots+\frac1{N-1}\ge \log N$.
In particular, $\AT(G,r)\ge \frac pr\cdot N\log N$ for an arbitrary graph $G$
and $\AT(A,r)\ge \frac{r-1}{r^2}\cdot N\log N $ for an arbitrary amplifier $A$.
\end{theorem}
\begin{proof}
Consider a modified Moran process $M'$ that is identical with the standard Moran process, except that if the mutation goes extinct then in the next step we again initialize a single mutant uniformly at random and continue the process. Clearly, the modified process $M'$ always terminates with the mutants fixating and its expected fixation time is given by $T'(G,r)=\frac1p\cdot \AT(G,r)$.

Given any subset $X$ of the vertices, let $p^X$ be the probability to gain a mutant in a single step from a configuration consisting of mutants at vertices of $X$ and residents elsewhere. 
To gain a mutant, one of the $|X|$ mutants has to be selected for reproduction and then the offspring has to replace a resident. The probability of the first event alone equals $\frac {r|X|}{N+(r-1)|X|}$, hence we get an upper bound
$$p^X\le \frac {r|X|}{N+(r-1)|X|} \le  \frac {r|X|}N  \equiv p_{|X|}$$
that doesn't depend on $X$ but only on $|X|$.

Finally, fix $k\in\{1,2,\dots,N-1\}$ and observe that any evolutionary trajectory in $M'$ has to, at some point, reach a state with $k$ mutants and gain another mutant from there. Hence, in expectation, the evolutionary trace spends at least $\frac 1{p_k}$ steps in states corresponding to configurations with $k$ mutants. By linearity of expectation, summing over $k$ gives
$$T(G,r)=p\cdot T'(G,r)\ge p\sum_{k=1}^{N-1} \frac N{r\cdot k} =\frac pr\cdot N\cdot H_{N-1} 
$$
as desired.
\end{proof}

Several remarks are in order.

First, we emphasize that the proof applies to all graphs, possibly containing directed edges, weighted edges, and/or self-loops.

Second, we note that the same proof goes through for any initialization scheme $\mathcal{S}$ (with $p=\fp(G,r)$ replaced by the fixation probability $p^{\mathcal{S}}$ under that initialization scheme $\mathcal{S}$). Specifically, it applies to temperature initialization and also to schemes in which the first mutant is initialized to a fixed vertex.

Third, we discuss relationship between absorption time and fixation time. Note that Theorem~\ref{thm:noaccelerators} provides a lower bound on the absorption time $\AT(G,r)$ which is a weighted average of the fixation time $T(G,r)$ and the extinction time $\ET(G,r)$. Since the evolutionary trajectories leading to extinction are typically shorter than those leading to fixation, the fixation time tends to be even longer than the absorption time. In fact, the inequality $T(G,r)>\AT(G,r)$ holds for all undirected graphs $G$ and all values $r>1$ that we tested.
On the other hand, there do exist directed graphs for which the opposite inequality $T(G,r)<\AT(G,r)$ holds.
As an example, consider $r=4$ and a graph $G$ consisting of three vertices $\{u,v_1,v_2\}$ and edges $\{u\to v_1, u\to v_2, v_1 \leftrightarrow v_2\}$. Then we easily check that $T(G,4)=3.25$ while $\AT(G,4)=19.25$.
In fact, in terms of fixation time, this graph $G$ is even slightly faster than the complete graph $K_3$, as we have  $T(K_3,4)=3+\frac47>3.25$.

%
%
\pagebreak
\section*{Appendix 2: Uniform initialization, $\eps$-Balanced bipartite graphs}\label{app:2}
\label{balanced}
In this subsection we analyze the $\eps$-Balanced bipartite graph $B_{N,\eps}$.
Recall that $B_{N,\eps}$ consists of $c=N^{1-\eps}$ vertices in the (smaller) center and $N$ vertices in the outside part, each two vertices from different parts connected by an edge.

We prove the following theorem.

\begin{theorem} Fix $\eps\in(0,1]$ and $r>1$. Let $B_{N,\eps}$ be the $\eps$-Balanced bipartite graph. Then
\begin{itemize}
\item $\fp(B_{N,\eps},r)\to 1-1/r^2$.
\begin{itemize}
\item (small center) If $\eps\in(0.5,1)$ then there exist constants $c_1$, $c_2$ such that 
$$c_1\cdot N^{1+\eps}\log N \le \AT(B_{N,\eps},r) \le c_2\cdot N^{1+\eps}\log N.$$
\item (large center) If $\eps\in(0,0.5)$ then there exist constants $c_1$, $c_2$ such that 
$$c_1\cdot N^{1+\eps}\log N \le \AT(B_{N,\eps},r) \le c_2\cdot N^{1+2\eps}\log N.$$
\end{itemize}
\end{itemize}
Moreover, the fixation time $ T(B_{N,\eps},r)$ satisfies the same inequalities.
\end{theorem}

\subsubsection*{Martingales background}
First, we recall the following facts about martingales (see~\cite{Monk2014}).
Fix $r>1$.
Given a complete bipartite graph with $v$ vertices at the outside part and $c$ vertices in the center, the state (configuration) space can be parametrized by the number $0\le i\le v$ of mutants in the outside part and the number $0\le j\le c$ of mutants in the center.
For each state $(i,j)$, let $\fp(i,j)$ be the fixation probability starting from that state. There is a formula for $\fp(i,j)$ which can be computed as follows: Let
$$h_v=\frac{v+cr}{vr^2+cr}, \qquad h_c=\frac{c+vr}{cr^2+vr}$$ 
and for every state $(i,j)$ define a potential function $\phi(i,j)=h_v^i\cdot h_c^j$.
(Note that $\phi(i+1,j)=\phi(i,j)\cdot h_v$ and $\phi(i,j+1)=\phi(i,j)\cdot h_c$.)
Then
$$\fp(i,j)=\frac{\phi(0,0)-\phi(i,j)}{\phi(0,0)-\phi(v,c)}=\frac{1-\phi(i,j)}{1-\phi(v,c)}.$$
For the rest of this section, we will be using these results for $c=N^{1-\eps}$ and $v=N$.

\subsubsection*{Fixation probability}
With the martingales background, the analysis of the fixation probability is relatively straightforward.
\begin{lemma}\label{l:bipprob}
Fix $\eps\in(0,1]$ and $r>1$. As $N\to\infty$, we have $\fp(B_{N,\eps},r)\to 1-1/r^2$.
\end{lemma}
\begin{proof} The original mutant appears at the outside part with probability $N/(N+N^{1-\eps})\to 1$. Since $\phi(1,0)=h_v\to 1/r^2$ and $\phi(v,c)=h_v^v\cdot h_c^c <h_v^v \to 0$ as $N\to\infty$, we compute 
$$\fp(B_{N,\eps},r)=\frac{1-\phi(1,0)}{1-\phi(v,c)}\to_{N\to\infty} 1-1/r^2.
$$
\end{proof}

\subsubsection*{Lower bound on fixation time}
Next, we present the lower bounds for the absorption and fixation time.
The idea is to consider the expected time $t_k$ to gain one mutant in the outside part, if there are currently $k$ mutants there. By bounding those times and summing up we obtain the following lemma.

\begin{lemma}\label{timelower} Fix $\eps\in(0,1]$. Then
$$\AT(B_{N,\eps},r)= \Omega(N^{1+\eps}\log N),\quad T(B_{N,\eps},r)= \Omega(N^{1+\eps}\log N).$$

\end{lemma}
\begin{proof}
For the absorption time, we proceed as in the proof of Theorem~\ref{thm:noaccelerators}, that is, we restart the process each time the mutants go extinct. The modified process $M'$ always terminates with the mutants fixating and its expected fixation time is given by $T'(B_{N,\eps},r)=\frac1{\fp(B_{N,\eps},r)}\cdot \AT(B_{N,\eps},r)$.

Consider a state with $1\le k\le N-1$ mutants in the outside part and $0\le j\le c$ mutants at the center. Let $F=N+c+(r-1)(j+k)>N$ be the total fitness of the population. The probability that in the next step we gain one mutant in the outside part equals
$$\frac{r\cdot j}F \cdot \frac{N-k}N\le \frac{r\cdot c}{N^2}\cdot (N-k)\equiv p_k.
$$
Since $p_k$ is independent of $j$, the expected time to reach some state with $k+1$ mutants, starting in any state with $k$ mutants in the outside part, is at least
$$ \frac1{p_k} = \frac1r\cdot\frac{N^2}c\cdot\frac1{N-k}\equiv t_k.
$$
In order to fixate, we need to pass through a state with $k$ mutants in the outside part, for each $k=1,\dots,N-1$. By linearity of expectation,
\begin{align*}
\AT(B_{N,\eps},r)&=\fp(B_{N,\eps},r)\cdot T'(B_{N,\eps},r) \\
&\ge \fp(B_{N,\eps},r)\cdot \sum_{k=1}^{N-1} \frac1r\cdot\frac{N^2}c\cdot\frac1{N-k} \to \frac{r^2-1}{r^3}\cdot N^{1+\eps}\cdot\sum_{k=1}^{N-1}\frac1k = \Theta(N^{1+\eps}\log N).
\end{align*}

For the fixation time, we perform a standard construction to obtain a different modified process $M''$ that only includes the trajectories that lead to fixation. Specifically, we remove the state $(0,0)$ (the only state $s$ with $\fp(s)=0$) and, for any two other states $s$ and $t$, we renormalize the transition probability $p(s\to t)$ to a new value $p''(s\to t)=p(s\to t)\cdot\frac{\fp(t)}{\fp(s)}$. It is a standard result that in this way we have constructed a Markov chain with only one absorbing state whose absorption time is equal to the fixation time of the original process, that is, $T(B_{N,\eps},r)=\AT''(B_{N,\eps},r)$. To get a lower bound for $\AT''(B_{N,\eps},r)$, we proceed as before.

Due to the renormalization, each $p_k$ ($k=1,\dots,N-1$) gets multiplied by a ratio of two fixation probabilities that can be upper bounded by $$\frac{\max_j\{\fp(k+1,j)\} }{\min_j\{\fp(k,j)\} }.$$ Note that for $k\ge 1$ the denominator is at least a constant (recall that $\fp(1,0)\to 1-1/r^2$ for large $N$), hence the ratio can be further upper bounded by $1/c_0$ for any $c_0<1-1/r^2$ and $N\to \infty$. Hence $t''_k= 1/p''_k \ge c_0/p_k$. This gives
\begin{align*}
T(B_{N,\eps},r)&=\AT''(B_{N,\eps},r) \ge
\sum_{k=1}^{N-1} t''_k \ge c_0 \sum_{k=1}^{N-1} \frac1r\cdot\frac{N^2}c\cdot\frac1{N-k} 
&= \frac{c_0}{r}\cdot N^{1+\eps}\cdot\sum_{k=1}^{N-1}\frac1k = \Theta(N^{1+\eps}\log N)
\end{align*}
as desired.
\end{proof}

\subsubsection*{Upper bound: ``small'' center}
For the upper bound, we distinguish two cases. First, we assume that $\eps\in(1/2,1]$, that is $c=o(\sqrt N)$.

The idea is to again work with the restarted process and moreover to split the set of states into \textit{sections} as follows: section $S_i$ consists of all the states with $i$ mutants in the outside part.
Then we consider a Markov chain $\Em'$ whose nodes are the sections $S_i$. By construction, the only transitions with nonzero probability are of the form $S_i\to S_{i\pm1}$ or $S_i\to S_i$. In the following sequence of Lemmas, we provide upper bounds for the expected number of transitions from $S_{i+1}$ to $S_i$ and for the expected number of transitions within each $S_i$. Summing up, we obtain an upper bound for the fixation time in the original Markov chain.

Formally, fix $i$ and let
\begin{itemize}
\item $f_{\max}=\max_j\{\fp(i,j)\}$ be the maximum fixation probability from a state in $S_i$. Clearly, $f_{\max}$ is attained in state $(i,c)$.
\item $g_{\min}=\min_j\{\fp(i+1,j)\}$ be the minimum fixation probability from a state in $S_{i+1}$. Clearly, $g_{\min}$ is attained in state $(i+1,0)$.
\item $q=\min_j \{q_j\}$ where $q_j$ is the probability that an evolutionary trajectory starting at $(i+1,j)$ fixates at $(v,c)$ before visiting any state in $S_i$.
\end{itemize}

First, since $\eps>1/2$ we have the following:
\begin{lemma}\label{l:sections} $h_c^c\to_{N\to\infty}1$ and $f_{\max}<g_{\min}$ (for large enough $N$)
\end{lemma}
\begin{proof}
We have
$$h_c^c\approx \left(1-\frac{r-1/r}{N^\eps}\right)^{(N^{1-\eps})}$$
For $N\to \infty$ we have $N^\eps\to\infty$. If the parenthesis was raised to power $N^\eps$, the limit would have been $\exp(-(r-1/r))$, a constant. Since $N^{1-\eps}=o(N^\eps)$ for $\eps>1/2$, we have $\lim_{N\to\infty} h_c^c=1$. Hence $h_c^c>h_v$, then $\phi(i,c)=h_c^c\cdot \phi(i,0)>h_v\cdot \phi(i,0)=\phi(i+1,0)$ and thus $f_{\max}=\fp(i,c)<\fp(i+1,0)=g_{\min}$ as desired.
\end{proof}

We aim to bound $q$ from below and use it to bound the expected number $X$ of transitions from (any state in) $S_{i+1}$ to (any state in) $S_i$ from above.

\begin{lemma}\label{l:qbyfprobs} $q \ge\frac{g_{\min}-f_{\max}}{1-f_{\max}}$
\end{lemma}
\begin{proof}
Let's run an evolutionary trajectory from some state $(i+1,j)$ in $S_{i+1}$. The trajectory can't go extinct without hitting $S_i$. Conditioning on if the trajectory first fixates or hits $S_i$, we can write
$$g_{\min}\le \fp(i+1,j) \le q_j \cdot 1 + (1-q_j)\cdot f_{\max}$$
which rewrites as $$q_j\ge\frac{g_{\min}-f_{\max}}{1-f_{\max}}.$$
This is true for every $j$, hence it is true for $q =\min_j \{q_j\}$ too.
\end{proof}

Let $X$ be a random variable counting the transitions from any state in $S_{i+1}$ to any state in $S_i$, starting from any state.
\begin{lemma}\label{l:xbyfprobs} $\E[X]\le \frac{1-q}q = \frac{1-g_{\min}} {g_{\min}-f_{\max}}.$
\end{lemma}
\begin{proof}
Any two transitions from section $S_{i+1}$ to section $S_i$ are necessarily separated by an intermediate visit to section $S_{i+1}$. Any time we are in section $S_{i+1}$, with probability at least $q$ we fixate before hitting section $S_i$ again. Hence 
$$ E[X]\le q\cdot 0 + (1-q)(1+E[X]).
$$
Rewriting and using the bound for $q$ we obtain
$$\E[X]\le \frac{1-q}q =\frac1q-1 = \frac{1-f_{\max}} {g_{\min}-f_{\max}} -1=  \frac{1-g_{\min}} {g_{\min}-f_{\max}}.
$$
\end{proof}

Rewriting $g_{\min}$ and $f_{\max}$ in terms of $h_v$, $h_c$ we deduce that $\E[X]$ is constant.
\begin{lemma}\label{l:xconst} $\E[X]\le \frac{1}{r^2-1}$ (for large enough $N$)
\end{lemma}
\begin{proof}
Recall that $\fp(i,j)=\frac{1-\phi(i,j)}{d}$ where $d=1-\phi(v,c)$ doesn't depend on $i$, $j$. Plugging this in the bound from Lemma~\ref{l:xbyfprobs} we get
\begin{align*}
\E[X]&\le \frac{1-g_{\min}} {g_{\min}-f_{\max}} =  \frac{1-\frac{1-\phi(i+1,0)}d } {\frac{1-\phi(i+1,0)}d-\frac{1-\phi(i,c)}d } \\
&= \frac{ d - (1-\phi(i+1,0))}{1-\phi(i+1,0) - (1-\phi(i,c))} < \frac{\phi(i+1,0)}{\phi(i,c)-\phi(i+1,0)}.
\end{align*}
Using the definition $\phi(i,j)=h_v^ih_c^j$ and dividing by $h_v^i$ this can be further rewritten as
$$\E[X]<\frac{\phi(i+1,0)}{\phi(i,c)-\phi(i+1,0)} = \frac{h_v}{h_c^c-h_v} \to_{N\to\infty} \frac{1/r^2}{1-1/r^2}
$$
as desired.
\end{proof}

Let $\E[\El_i]$ be the expected number of ``looping'' transitions of the form $S_i\to S_i$ before a transition of the form $S_i\to S_{i\pm1}$ occurs (or the process reaches an absorbing state). The following lemma bounds $\E[\El_i]$ from above.

\begin{lemma}\label{l:loops} For $i=1,2,\dots,N-1$ we have $\E[\El_i]\le \frac{r\cdot N(N+c)}{c\cdot \min\{i,N-i\}}-1$.
Moreover, $\E[\El_0]\le r(N+c)-1$ and $\E[\El_N]\le r(N+c)-1$.
\end{lemma}
\begin{proof} Crudely (not caring about $r$). First, let $i=1,\dots,N-1$. We pick a vertex in the center with probability at least $\frac{1\cdot c}{r(N+c)}$. No matter its type, there are at least $\min\{i,N-i\}$ vertices of the other type at the outside part. Hence with probability
$$p\ge \frac{1\cdot c}{r(N+c)} \cdot \frac{\min\{i,N-i\}}N$$
we transition to section $S_{i\pm 1}$ in one step.  As before, we get the result from
$\E[\El_i]\le \frac1p-1$.
Second, if $i=0$ or $i=N$ and we are not in an absorbing state then there exists a vertex in the center whose type is different to the type of all vertices in the outside part. Hence $p\ge \frac1{r(N+c)}$ and we conclude as in the first case.
\end{proof}

We are ready to sum those contributions up.

\begin{lemma}\label{smallcenter} If $\eps\in(1/2,1]$, $c=N^{1-\eps}=o(\sqrt N)$ and $r>1$ then $$\AT(B_{N,\eps},r) = O\left(\frac{N^2}c\cdot\log N\right)=O(N^{1+\eps}\log N).$$
\end{lemma}
\begin{proof} 
As in the proof of Theorem~\ref{thm:noaccelerators} we restart the process anytime the mutants fixate.
Consider the one-dimensional Markov chain $\Em'$ whose vertices are the sections $S_i$, $i=0,\dots,N$. Fix $i\in\{1,\dots,N-1\}$ and let $f(r)=\frac{1}{r^2-1}$. On average, there are at most $f(r)$ transitions $S_{i+1}\to S_i$. Also, on average there are at most $f(r)$ transitions $S_i\to S_{i-1}$, hence there are at most $f(r)+1$ transitions $S_{i-1}\to S_i$ for a total of at most $2f(r)+1$ transitions from outside of $S_i$ to $S_i$. Similarly, on average there are at most $f(r)$ transitions into $S_0$ and at most $f(r)+1$ transitions into $S_N$. Every time there is a transition into $S_i$, there are on average $\E[\El_i]$ transitions within $S_i$. By linearity of expectation, the total expected number of transitions is at most
\begin{align*}
\AT(B_{N,\eps},r) &= \fp(B_{N,\eps},r)\cdot T'(B_{N,\eps},r) \\
 &\le \fp(B_{N,\eps},r)\cdot\left( f(r)\cdot r(N+c)+ \left(\sum_{i=1}^{N-1} (2f(r)+1)(1+\E[\El_i])\right) +(f(r)+1)\cdot r(N+c) \right)\\
  &= \fp(B_{N,\eps},r)\cdot (2f(r)+1)\cdot r(N+c)\cdot \left(1+\frac Nc \sum_{i=1}^{N-1}\frac1{\min\{i,N-i\}}\right) =  \Theta\left(\frac{N^2}c\cdot\log N\right),
\end{align*}
where the last equality follows from the sum being $\Theta(2\log(N/2))=\Theta(\log N)$ and from $c=o(N)$.
\end{proof}

\subsubsection*{Upper bound: ``large'' center}
Note that the argument used for small center fails for $\eps\le 1/2$ because the difference $g_{\min}-f_{\max}$ becomes zero or even negative. Indeed, for $\eps= 1/2$ we have $h_c^c\to_{N\to\infty} \exp(-(r-1/r))<1/r^2$
 and for $\eps<1/2$ the inequality is even stricter.
However an analogous argument can be made to work if we split the state space into different ``tilted'' sections, taking $\eps$ into account.
The idea of the proof is that we fix $\eps\in(0,1/2)$, consider large $N$, and look at a complete bipartite graph $B_{N,\eps}$. We assume that $r$ is such that there exists an integer $t$ called ``tilt'' satisfying $h_c^t=h_v$.
This assumption guarantees that the states $(i,j+t)$ and $(i+1,j)$ are assigned exactly the same potential.
We can then split the state (configuration) space into $\Theta(N)$ sections where each section is not a vertical line but a set of $c$ states that looks like a line tilted with slope $-t$ (see figure). We then proceed as before, providing an upper bound for the number of transitions across sections and within sections. The result follows by summing up.

\begin{figure}[h] 
  \centering
  \includegraphics[scale=1]{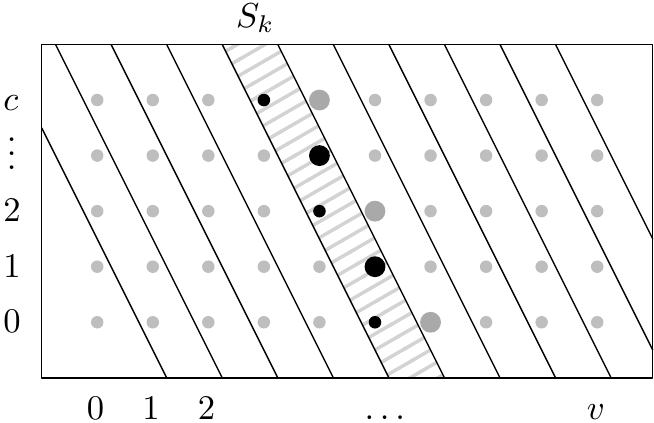}
      
\caption{\footnotesize{\textbf{Tilted sections of $B_{N,\eps}$.} For $\eps\in(0,1/2)$ we split the state (configuration) space into $\Theta(N)$ ``tilted'' sections $S_k$. Here the tilt is $t=2$. The maximum potential within $S_k$ is attained at any thick black vertex, the minimum potential within $S_{k+1}$ is attained at any thick grey vertex.}}
\label{fig:biptime}
\end{figure}

In the rest of the section we formalize this idea.
First, we define the (tilted) sections. Let $s=v+\lfloor c/t\rfloor$ and for $k=0,\dots,s$ let
$$S_k=\left\{(i,j) : i+\lfloor j/t\rfloor = k \right\}.$$
As before, we fix $k$ and define
\begin{itemize}
\item $f_{\max}=\max \{\fp(i,j) : (i,j)\in S_k\}$,
\item $g_{\min}=\min \{\fp(i,j) : (i,j)\in S_{k+1}\}$, and
\item $q=\min_{(i,j)\in S_{k+1}} \{q_{(i,j)}\}$ where $q_{(i,j)}$ is the probability that an evolutionary trajectory starting at state $(i,j)$ belonging to $S_{k+1}$ fixates at $(v,c)$ before visiting any state in $S_k$.
\end{itemize}
Clearly, $f_{\max}$ is attained for any ``top'' state of $S_k$ within its column (possibly not in the ``top'' row).
Similarly, $g_{\min}$ is attained in state ``bottom'' state of $S_{k+1}$ (possibly not in the ``bottom'' row).
Note that by construction, those two states are assigned potentials that differ by a factor of $h_c$.

As before, let $X$ be a random variable counting the transitions from any state in $S_{k+1}$ to any state in $S_k$, starting from any state. The following lemma bounds $\E[X]$ from above. Note that this time the bound is super-constant.
\begin{lemma}\label{l:xbounded} $\E[X]\le \frac{r}{r^2-1}\cdot N^\eps$ (for large $N$)
\end{lemma}
\begin{proof} Note that Lemma~\ref{l:qbyfprobs} and Lemma~\ref{l:xbyfprobs} are valid for tilted sections too. Let $(i,j)$ be some state in $S_k$ for which the value $f_{\max}$ is attained. Then $g_{\min}$ is attained at a state whose potential is equal to $h_c\cdot\phi(i,j)$. We continue as in the proof of Lemma~\ref{l:xconst} to get
$$\E[X]\le \frac{1-g_{\min}} {g_{\min}-f_{\max}} = \frac{ h_c\cdot \phi(i,j)-\phi(v,c) }{ \phi(i,j)-h_c\cdot \phi(i,j) } < \frac{h_c}{1-h_c}.
$$
Since for large $N$ we have  $h_c\approx 1-(1-1/r)/N^\eps$, the right-hand side can be approximated as
$$
\frac{h_c}{1-h_c}<\frac1{1-h_c}= \frac{N^\eps}{r-1/r}=\frac{r}{r^2-1}\cdot N^{\eps}.
$$
\end{proof}

It remains to bound the expected number $\E[\El_k]$ of the ``looping'' transitions of the form $S_k\to S_k$ before a transition of the form $S_k\to S_{k\pm1}$ occurs. This is done as before, observing that any two states that differ only in the number of mutants in the outside part of the graph always lie in different sections. Hence Lemma~\ref{l:loops} holds.

Finally, we prove the last inequality in Theorem~\ref{biptime}.
\begin{lemma}\label{largecenter} If $\eps\in(0,1/2)$ then $$\AT(B_{N,\eps},r) = O( N^{1+2\eps}\log N).$$
\end{lemma}
\begin{proof}
As before, let $f(r)=\frac{r}{r^2-1}\cdot N^{\eps}$. By linearity of expectation, the total expected number of transitions is
\begin{align*}
\AT(B_{N,\eps},r) &\le  \fp(B_{N,\eps},r)\cdot\left( f(r)\cdot r(N+c)+ \left(\sum_{i=1}^{N-1} (2f(r)+1)(1+\E[\El_i])\right) +(f(r)+1)\cdot r(N+c)  \right)\\
  &=  \fp(B_{N,\eps},r)\cdot (2f(r)+1)\cdot r(N+c)\cdot \left(1+\frac Nc \sum_{i=1}^{N-1}\frac1{\min\{i,N-i\}}\right) =
  O\left(N^{\eps}\cdot N\cdot \frac Nc\cdot \log N\right)\\
  &=O(N^{1+2\eps}\log N)
 \end{align*}
and the result follows.
\end{proof}

Finally, we observe that an upper bound on $\AT(B_{N,\eps},r)$ immediately implies an asymptotically matching upper bound on $T(B_{N,\eps},r)$.
\begin{lemma}\label{abs2fix} Fix $r>1$ and $\eps>0$. If $\AT(B_{N,\eps},r) = O( N^\alpha \log N)$ then $T(B_{N,\eps},r) = O( N^\alpha \log N)$
\end{lemma} 
\begin{proof} Since the absorption time is a weighted average of the fixation time and the extinction time, we have
$$
\AT(B_{N,\eps},r) =\fp(B_{N,\eps},r) \cdot T(B_{N,\eps},r)  + (1-\fp(B_{N,\eps},r))\cdot \ET(B_{N,\eps},r) )\ge \fp(B_{N,\eps},r) \cdot T(B_{N,\eps},r)
$$
Since $\fp(B_{N,\eps},r)\to 1-1/r^2$ as $N\to \infty$ and $\AT(B_{N,\eps},r) = O( N^\alpha \log N)$, the result follows.
\end{proof}
Altogether, Lemmas~\ref{l:bipprob},~\ref{timelower},~\ref{smallcenter},~\ref{largecenter} and~\ref{abs2fix} prove all the statements of Theorem~\ref{biptime}.

%
%
\pagebreak
\section*{Appendix 3: Temperature initialization, $\eps$-Weighted bipartite graphs}\label{app:3}
In this section we analyze the Weighted bipartite graphs $W_{N,\eps}$ under temperature initialization.
Recall that $W_{N,\eps}$ is a complete bipartite graph with one (smaller) part of size $c=N^{1-\eps}$, one (larger) part of size $N$, and every vertex of the larger part having a self-loop of such a weight $w$ that $N^{-\eps/2}=\frac{N^{1-\eps}}{w+N^{1-\eps}}$.

We prove the following theorem.
\begin{theorem} Fix $\eps\in(0,1)$ and $r>1$. Let $W_{N,\eps}$ be the Weighted bipartite graph. Then
\begin{itemize}
\item $\fp(W_{N,\eps},r)\to 1-1/r^2$.
\item There exist constants $c_1$, $c_2$ such that 
$$c_1\cdot N^{1+\eps}\log N \le \AT(B_{N,\eps},r) \le c_2\cdot N^{1+\frac32\eps}\log N.$$
\end{itemize}
Moreover, the fixation time $ T(B_{N,\eps},r)$ satisfies the same inequalities.
\end{theorem}

\subsubsection*{Martingales for Weighted bipartite graphs}
First, we recall more martingales background.

Fix $r>1$.
Given integers $v$, $c$, and a real number $q\in(0,1)$, let $W(c,v,q)$ be a Weighted complete bipartite graph with $c$ vertices in the smaller part (center) and $v$ vertices at the larger (outside) part, each of them with an extra self-loop of such weight $w$ that $q=\frac{c}{w+c}$ is the probability that when a vertex in the larger part is selected for reproduction, its offspring replaces one of the vertices in the smaller part (as opposed to replacing its parent via the self-loop).
Then, as with the unweighted complete bipartite graphs, the state space can be parametrized by the number $0\le i\le v$ of mutants in the outside part and the number $0\le j\le c$ of mutants in the center and the fixation probabilities from all the states can be computed similarly to above, with $v$ replaced by $v\cdot q$.

Namely, let
$$h_v=\frac{qv+cr}{qvr^2+cr}, \qquad h_c=\frac{c+qvr}{cr^2+qvr}$$ 
and for every state $(i,j)$ of $i$ mutants in the outside part and $j$ mutants in the center, define a potential function $\phi(i,j)=h_v^i\cdot h_c^j$.
Then we easily check the fixation probability from a state $(i,j)$ is given by
$$\fp(i,j)=\frac{\phi(0,0)-\phi(i,j)}{\phi(0,0)-\phi(v,c)}=\frac{1-\phi(i,j)}{1-\phi(v,c)}.$$

\subsubsection*{Fixation probability}
With the extra martingales background, the analysis of the fixation probability is again relatively straightforward.
\begin{lemma}\label{weightedprob} $\fp_\temp (W_{N,\eps},r)\to 1-1/r^2$.
\end{lemma}
\begin{proof} The first mutant is introduced in the center with probability proportional to $q\cdot v$ and to the outside part with probability proportional to $c+(1-q)v$. Since 
 $q=o(1)$, it is introduced to the outside part with high probability.
The fixation probability $\fp(1,0)$ starting from a state with a single mutant in the outside part satisfies
$$ \fp(1,0)=\frac{1-\phi(1,0)}{1-\phi(v,c)}=\frac{1-h_v}{1-h_v^vh_c^c}.
$$
Since $h_v\approx \frac1{r^2}$ and $h_v^v\to 0$ as $N\to\infty$, we have $\fp(1,0)\to 1-1/r^2$.
\end{proof}

\subsubsection*{Fixation time}
The arguments for fixation time are direct translations of arguments for (unweighted) $\eps$-Balanced bipartite graphs (see Section~\ref{balanced}). For the lower bound, Lemma~\ref{timelower} still applies. For the upper bound, we proceed analogously.

\begin{lemma}\label{weightedtime} If $\eps\in(0,1)$ and $r>1$ then $T_\temp(W_{N,\eps},r)=O( N^{1+\frac32\eps}\log N)$.
\end{lemma}
\begin{proof}
Fixing $k$ and considering the section $S_k$, we denote by $X$ the expected number of transitions from any state in $S_{k+1}$ to any state in $S_k$. As in Lemmas~\ref{l:xconst} and \ref{l:xbounded} we get
$$\E[X]\le \frac{h_c}{1-h_c} =\Theta(N^{\eps/2}).$$
Lemma~\ref{l:loops} then yields
$$ T_\temp(W_{N,\eps})=O\left(N^{\eps/2}\cdot N\cdot \frac Nc\cdot \log N\right)=O(N^{1+\frac32\eps}\log N).
$$
\end{proof}
Altogether, Lemmas~\ref{timelower},~\ref{abs2fix},~\ref{weightedprob} and~\ref{weightedtime} prove all the statements of Theorem~\ref{weighted}

%
%
\pagebreak
\section*{Appendix 4: Time on Complete graph}\label{app:4}

For reference purposes we compute the absorption time $\AT(K_N,r)$, the fixation time $T(K_N,r)$, and the extinction time $\ET(K_N,r)$ of a single advantageous mutant ($r>1$) on a Complete graph $K_N$, using the standard difference method.
\begin{theorem} Fix $r>1$ and let $K_N$ be the Complete graph on $N$ vertices. Then
\begin{align*}
\AT(K_N,r)&=(N-1)H_{N-1}\cdot \frac{r+1}r + (N-1)\cdot\log(1-1/r) - \frac1{r(r-1)} +o(1),\\
T(K_N,r)&=(N-1)H_{N-1}\cdot \frac{r+1}{r-1} + (N-1)\cdot\frac{r+1}{r-1}\log(1-1/r)+o(N),\\
\ET(K_N,r)&=  (N-1)\cdot \log\left(\frac r{r-1}\right) +o(N).
\end{align*}
In particular for $r=1+s$, $s>0$ small, we have  $\AT(K_N,r)\approx 2\cdot N\log N$, $T(K_N,r)\approx\frac 2s\cdot N\log N $, and $\ET(K_N,r)\approx \frac1s\cdot N $.
\end{theorem}
\begin{proof}
First, we compute the absorption time, then the fixation time and finally the extinction time.

%
\sne{Absorption time} Fix $N$ and $r$ and for $k=0,\dots,N$ let $T_k$ be the expected absorption time from a state with $k$ mutants. Clearly $T_0=T_N=0$ and for $k=1,\dots,N-1$ we have
$$ T_k= 1+ p(k,k)T_k + p(k,k-1)T_{k-1} + p(k,k+1)T_{k+1},
$$
where $p(i,j)$ is the transition probability from a state with $i$ mutants to a state with $j$ mutants.
Specifically, we have
$$ p(k,k-1)=\frac{N-k}{N+(r-1)k}\cdot \frac{k}{N-1}.
\qquad \textrm{and}\qquad
p(k,k+1)=\frac{r\cdot k}{N+(r-1)k}\cdot \frac{N-k}{N-1}.
$$
Plugging in those values of $p(i,j)$, the above equation can be rewritten as 
$$ T_{k+1}-T_k= \frac1r\left(T_k-T_{k-1}\right) - \left(\frac{N-1}{rk} +\frac{N-1}{N-k} \right).
$$
Setting $\Delta_k\equiv T_k-T_{k-1}$ and $x_k=\frac{N-1}{rk} +\frac{N-1}{N-k}$ this further rewrites as
$$ \Delta_{k+1}= \frac1r\Delta_k - x_k.
$$
Specifically, we have $\Delta_1=T_1-T_0=T_1$ and $\Delta_1+\dots+\Delta_N=T_N-T_0=0$.
Let's write
\begin{align*}
\Delta_2 &= \frac1r \Delta_1 - x_1 ,\qquad\qquad(1)\\
\Delta_3 &= \frac1r \Delta_2 - x_2 ,\qquad\qquad(2)\\
\Delta_4 &= \frac1r \Delta_3 - x_3 ,\qquad\qquad(3)\\
&\dots \\
\Delta_N &= \frac1r \Delta_{N-1} - x_{N-1},\qquad(N-1)
\end{align*}

We aim to express each $\Delta_k$ in terms of $\Delta_1$ only.
Summing up $\frac1r(1)+(2)$ gives
$$\Delta_3=\frac1{r^2}\Delta_1 - \left(x_2+\frac1rx_1\right)
$$
Similarly, summing up $\frac1{r^2}(1)+\frac1{r}(2)+(3)$ gives
$$\Delta_4=\frac1{r^3}\Delta_1 - \left(x_3+\frac1rx_2+\frac1{r^2}x_1\right)
$$
and similarly all the way up to 
$$\Delta_N=\frac1{r^{N-1}}\Delta_1 - \left(x_{N-1}+\frac1rx_{N-2}+\dots+\frac1{r^{N-2}}x_1\right).
$$

Summing up all of them, together with an extra equation $\Delta_1=\Delta_1$, we get
$$0=\Delta_1+\dots+\Delta_N = \Delta_1\left(1+\dots+\frac1{r^{N-1}}\right)
  - \big( x_1(1+\dots+1/r^{N-2}) +x_2(1+\dots+1/r^{N-3}) +\dots + x_{N-1}\cdot 1  \big)
$$
and in turn
\begin{align*}
\Delta_1 &=\frac{1-1/r}{1-1/r^N}\cdot \frac{ x_1(1-1/r^{N-1}) + x_2(1-1/r^{N-2})+\dots+x_{N-1}(1-1/r) } {1-1/r} \\
 &=\frac1{1-1/r^N} \Big( \underbrace{ \sum_{k=1}^{N-1} x_k }_{A} - \underbrace{ \sum_{k=1}^{N-1}\frac{x_k}{r^{N-k}} }_{ B} \Big).
\end{align*}
For $A$ we easily get $A= \frac{r+1}r (N-1)H_{N-1}$. For $B$, we plug in $x_k=\frac{N-1}{rk} +\frac{N-1}{N-k}$, split $B=B_1+B_2$ and separately compute the sums using a well-known limit
$$ B_2=(N-1)\sum_{i=1}^{N-1}\frac1{N-k}\cdot \frac1{r^{N-k}}=(N-1)\left(\log(1-1/r)+O(1/r^N)\right)\to (N-1)\log(1-1/r) +o(1)
$$
and an approximation
$$B_1=\frac1r \sum_{k=1}^{N-1}\frac{N-1}{N-k} \frac1{r^{k}}  = \frac1r\sum_{k=1}^{N-1}\frac1{r^k}  +  E(N) = \frac1{r(r-1)} +o(1)
$$
whose error term
$$E(N) = \sum_{k=1}^{N-1} \frac{k-1}{N-k}\cdot \frac 1{r^k}
$$
tends to 0, because the sum $S_1$ over the first $\sqrt[3] N$ terms satisfies
$$S_1\le \sqrt[3] N\cdot \frac{\sqrt[3] N}{N-\sqrt[3] N}\cdot \frac1{r} < \frac1{\sqrt[3] N}\cdot \frac 1r \to 0
$$
and the sum $S_2$ of the remaining terms satisfies
$$ S_2\le (N-\sqrt[3] N)\cdot \frac N1 \cdot \frac 1{r^{\sqrt[3] N}} \to 0.
$$
This concludes the proof of the absorption time.

%
\sne{Fixation time} We proceed similarly. 

As before, we fix $N$ and $r$ and for $k=1,\dots,N$ we let $\CT_k$ be the expected (conditional) fixation time from a state with $k$ mutants (for $k=0$ we define $\CT_0=0$). Then $\CT_N=0$ and for $k=1,\dots,N-1$ we have
$$ \fp_k\CT_k= \fp_k+ p(k,k)\cdot \fp_k\CT_k + p(k,k-1)\cdot \fp_{k-1}\CT_{k-1} + p(k,k+1)\cdot \fp_{k+1}\CT_{k+1},
$$
where $\fp_i=\frac{1-1/r^i}{1-1/r^N}$ are the fixation probabilities and $p(i,j)$ are the transition probabilities.
Setting $\Delta_k\equiv\fp_k\CT_k - \fp_{k-1}\CT_{k-1}$ and $x_k\equiv\frac{\fp_k}{p(k,k+1)}$ this can be rewritten as
$$ \Delta_{k+1}= \frac1r\Delta_k - x_k.
$$
Specifically, we have $\Delta_1=\fp_1\cdot\CT_1-\fp_0\cdot\CT_0=\fp_1\cdot\CT_1$ and $\Delta_1+\dots+\Delta_N=\fp_N\cdot\CT_N-\fp_0\cdot\CT_0=0-0=0$. As before, we obtain
$$(1-1/r^N)\cdot\Delta_1 =\underbrace{ \sum_{k=1}^{N-1} x_k }_{A} - \underbrace{ \sum_{k=1}^{N-1}\frac{x_k}{r^{N-k}} }_{ B}.
$$
This time, $p(k,k+1)=\frac{rk}{N+(r-1)k}\cdot\frac{N-k}{N-1}$ and thus
$$x_k=\frac{\fp_k}{p(k,k+1)}=\frac{1-1/r^k}{1-1/r^N}\cdot \left( \frac{r-1}r + \frac{N}{rk} \right) \cdot\frac{N-1}{N-k}
$$
and
$$A=\sum_{k=1}^{N-1} x_k=\frac{N-1}{1-1/r^N}\cdot \sum_{k=1}^{N-1} (1-1/r^k)\cdot \left(\frac{1}{N-k} + \frac{1}{rk}\right).
$$
Multiplying out the two parentheses we get
$$X\equiv\sum_{k=1}^{N-1} \frac{1}{N-k} + \frac{1}{rk} =(1+1/r)\cdot H_{N-1}
$$
and
$$Y\equiv\sum_{k=1}^{N-1} \frac1{r^k(N-k)} \to 0, \quad Z\equiv \sum_{k=1}^{N-1} \frac1{r^k\cdot rk} \to \frac1r\log(1-1/r).
$$
Hence
$$A=\frac{N-1}{1-1/r^N}\cdot(X+Y+Z) = (1+1/r)\cdot (N-1)H_{N-1} +  \frac1r\log(1-1/r)\cdot (N-1) + o(N).
$$
We proceed with $B$ analogously. This time, the only combination that survives is
$$\sum_{k=1}^{N-1} \frac1{r^{N-k}}\cdot \frac1{N-k} \to \log(1-1/r),
$$
hence $B = \log(1-1/r)\cdot (N-1) +o(N)$.

In total, we get
$$(1-1/r^N)\cdot \frac{1-1/r}{1-1/r^N}\cdot \CT_1 =(1-1/r^N)\cdot \Delta_1 = A+B = \frac{r+1}r\cdot (N-1)H_{N-1} + \frac{r+1}r\log(1-1/r)\cdot (N-1) + o(N)
$$
and finally the desired
$$\CT(N,r)=\CT_1 = \frac{r+1}{r-1}\cdot (N-1)H_{N-1} + \frac{r+1}{r-1}\log(1-1/r)\cdot (N-1) +o(N).
$$

%
\sne{Extinction time} A formula for the extinction time follows easily from the absorption time and the fixation time.

It suffices to note that $\AT_1=\fp_1\cdot T_1+(1-\fp_1)\cdot \ET_1$ and plug in the expressions for $\AT_1$ and $T_1$. The $N\log N$ term cancels out and we are left with
$$\ET_1= -\log(1-1/r)\cdot (N-1) +o(N).
$$
\end{proof}

\end{document}